\documentclass[6pt]{amsart}
\usepackage{graphicx}
\usepackage{moreverb}
\usepackage{booktabs}
\usepackage{placeins}
\usepackage{natbib}        
\usepackage{amsmath} 
\usepackage{mathrsfs}
\usepackage[colorlinks=true, citecolor=black, linkcolor=black, urlcolor=black]{hyperref} 
\usepackage{tikz, tikz-3dplot}
\usetikzlibrary{calc,angles,quotes}
\tdplotsetmaincoords{70}{110}
\usepackage{bbm}
\usepackage{thm-restate}
\usepackage{bm}
\usepackage{natbib} 
\usepackage{booktabs,tabularx, colortbl}
\usepackage{amsmath}
\usepackage{mathtools}
\usepackage{dsfont}
\usepackage{textcomp}
\usepackage{xcolor}
\usepackage[english]{babel}
\newcommand{\RNum}[1]{\uppercase\expandafter{\romannumeral #1\relax}}

\DeclareMathOperator{\var}{\mathbb{V}\mathrm{ar}}
\DeclareMathOperator{\cov}{\mathbb{C}\mathrm{ov}}
\DeclareMathOperator{\cor}{\mathbb{C}\mathrm{orr}}
\DeclareMathOperator{\Pe}{\mathbb{P}}
\DeclareMathOperator{\E}{\mathbb{E}}
\DeclareMathOperator{\N}{\mathcal{N}}
\newcommand{\indep}{\perp\!\!\!\perp}

\newcommand{\eps}{\varepsilon}

\newcommand{\R}{\mathbb{R}}
\newcommand{\X}{\mathbb{X}}
\newcommand{\A}{\mathbb{A}}
\newcommand{\Y}{\mathbb{Y}}
\newcommand{\W}{\mathbb{W}}

\newcommand{\He}{\mathcal{H}}
\renewcommand{\L}{\mathcal{L}}

\renewcommand{\P}{\mathbb{P}}

\newcommand{\bimp}{\hspace{1mm} \Leftrightarrow \hspace{1mm}}

\usepackage{xargs}
\newcommandx{\pto}[0]{\overset{P}{\to}}

\newcommand\BibTeX{{\rmfamily B\kern-.05em \textsc{i\kern-.025em b}\kern-.08em
T\kern-.1667em\lower.7ex\hbox{E}\kern-.125emX}}

\usepackage{indentfirst,csquotes}

\usepackage{amssymb,amsthm,amsmath}
\usepackage{xcolor,paralist,hyperref,fancyhdr,etoolbox}
\newtheorem{theorem}{Theorem}[]

\newtheorem{lemma}[theorem]{Lemma}

\newtheorem{corollary}[theorem]{Corollary}

\hypersetup{ colorlinks=true, linkcolor=black, filecolor=black, urlcolor=black }

\usepackage{lipsum}

\begin{document}

\title{Powering RCTs for marginal effects with GLMs using prognostic score adjustment} 
\author{Emilie Højbjerre-Frandsen*}

\author{Mark J. van der Laan}

\author{Alejandro Schuler}

\date{\today}


\let\thefootnote\relax

\begin{abstract}
In randomized clinical trials (RCTs), the accurate estimation of marginal treatment effects is crucial for determining the efficacy of interventions. Enhancing the statistical power of these analyses is a key objective for statisticians. The increasing availability of historical data from registries, prior trials, and health records presents an opportunity to improve trial efficiency. However, many methods for historical borrowing compromise strict type-I error rate control. Building on the work by \citet{Schuler2020} on prognostic score adjustment for linear models, this paper extends the methodology to the plug-in analysis proposed by \citet{RosenblumvanderLaan+2010} using generalized linear models (GLMs) to further enhance the efficiency of RCT analyses without introducing bias. Specifically, we train a prognostic model on historical control data and incorporate the resulting prognostic scores as covariates in the plug-in GLM analysis of the trial data. This approach leverages the predictive power of historical data to improve the precision of marginal treatment effect estimates. We demonstrate that this method achieves local semi-parametric efficiency under the assumption of an additive treatment effect on the link scale. We expand the GLM plug-in method to include negative binomial regression. Additionally, we provide a straightforward formula for conservatively estimating the asymptotic variance, facilitating power calculations that reflect these efficiency gains. Our simulation study supports the theory. Even without an additive treatment effect, we observe increased power or reduced standard error. While population shifts from historical to trial data may dilute benefits, they do not introduce bias.
\end{abstract} 

\maketitle

\markright{Powering RCTs for marginal effects with GLMs using prognostic score adjustment}

\section{Introduction}\label{sec1}

In this study, we will specifically focus on randomized clinical trials (RCTs), which form the cornerstone of regulatory approval processes, as seen in \citet{Hatswell}. Within RCTs, participants are randomly assigned to treatment and control groups, to gather data enabling the estimation of an intervention's effectiveness compared to a placebo or an active comparator. Randomization ensures that these groups are statistically similar in terms of observed and unobserved confounders. Thus, reducing potential biases and often eliminating the causal gap between the causal and statistical estimand. This process plays a pivotal role in facilitating impartial and sound decision-making when determining the causal impact of an intervention \cite{pearl}.

A substantial increase in the number of participants in clinical trials can lead to significant economic costs and prolonged timelines. This was demonstrated in the review by \citet{bentley} on the costs, impact, and value of conducting clinical trials. Having fewer participants in a trial offers a means to accelerate the trial process and reduce the overall cost of the trial. Additionally, recruitment itself poses a major challenge in various disease areas, and in certain scenarios, it may not be ethically suitable to conduct large placebo-controlled studies, as outlined by \citet{temple}. 

Adjusting for measured baseline covariates during the statistical analysis of a clinical trial has been shown to reduce the variance of treatment effect estimators, ultimately enhancing study power. Various studies have indicated that the inclusion of baseline variables reduces variance in linear models \cite{Tsiatis2008, MoorevanderLaan+2009, RosenblumvanderLaan+2010, Balzer2016, HojbjerreFrandsen2024}. In the case of binary outcomes, \citet{MoorevanderLaan+2009} illustrated the potential for precision gains through adjustment using logistic regression to estimate treatment effects on different scales, including risk difference, risk ratio, and odds ratio. These findings were extended by \citet{RosenblumvanderLaan+2010} to a broad range of generalized linear models (GLMs) for marginal effect estimation within RCTs. However, in the context of smaller clinical trials, incorporating numerous covariates can result in overfitting and inflated type \RNum{1} error rates. Additionally, the process of post hoc selection for the adjustment set raises concerns about potential data dredging, where researchers may (knowingly or unknowingly) seek covariates that maximize power, consequently further risking an increase in type \RNum{1} error rates. To address these concerns, Regulatory agencies such as the U.S. Food and Drug Administration (FDA) and the European Medicines Agency (EMA) have recommended the inclusion of only a select few highly prognostic baseline covariates, prespecified in the protocol or statistical analysis plan (SAP) prior to any data unblinding \cite{FDA_cov, EMA_cov}. For example, it is recommended to always include stratification variables despite the fact that the benefits of covariate adjustment are only realized when the covariates are predictive of the outcome and are imbalanced between the groups \cite{Moore2011}. While some emphasize the need for cautious consideration when performing adjusted analysis, others have focused on the gains in precision and advocate for more refined recommendations for adjustment \cite{Kahan2014, AUSTIN2010142}. To address the issue of adaptively selecting an appropriate model while still pre-specifying the analysis, \citet{Balzer2016, Balzer2024} have proposed the use of the Adaptive Pre-Specification (APS) method in conjunction with Targeted Maximum Likelihood Estimation (TMLE). This method has exhibited notable increases in study power. An alternative approach to further increase study power using covariate adjustment involves data fusion methods, which integrate data from past trials or real-world evidence (historical data) into the new RCT study. 

{\color{black}
In recent years, there has been a surge in interest in the concept of a "super covariate" or "prognostic covariate" \cite{Holzhauer2022}. 
Since at least \citet{tukey1993tightening} it has been recognized that power gains from covariate adjustment depend on how predictive the covariates are of the outcome. \citet{tukey1993tightening} also suggests that covariates might be merged into a ``super-covariate''.
\citet{CUPED} built on this idea, using historical data to learn how to construct this composite so that it is optimally predictive of the outcome.
Borrowing the term ``prognostic score'' from \citet{Hansen2008-nj}, \citet{Schuler2020} independently developed the same idea, added a power calculation and established local efficiency results. \citet{karlsson2024}, \citet{Holzhauer2022}, \citet{Liao2023} and others have built on these results.
}

Incorporating such a prognostic covariate does not disrupt the randomization of the trial, thereby avoiding susceptibility to bias that other data fusion methods may encounter \cite{Lim2018, hill, kaplan, Li2023}, even when Pocock's criteria are met \cite{Pocock1976}.
For many external control methods, additional assumptions like exchangeability of the potential outcome mean between the historical and current trial data are often needed to obtain an efficiency gain. If such an assumption were met, one would then augment the control arm with external controls as suggested by \citet{Li2023}. However, in practice the assumption of exchangeability is overly stringent. When using a prognostic score, however, there is no risk of increasing the type \RNum{1} error even when this assumption is not fulfilled. The same robustness is not guaranteed for many external control methods.

\citet{Schuler2020} describe and theoretically validate linear adjustment with a prognostic score\footnote{Also trademarked as PROCOVA\texttrademark\, by the company \href{https://www.unlearn.ai}{Unlearn.ai}.} for the estimation of the average treatment effect (ATE). Practical implementations and extensive simulation studies for this method have been conducted by \citet{HojbjerreFrandsen2024}. Using linear prognostic score adjustment with a pre-specified power, \citet{Schuler2020} demonstrated large reductions in the control arm size for phase \RNum{3} trials. In September 2022, the Committee for Medicinal Products for Human Use \cite{EMA} at the EMA issued a qualification opinion for prognostic adjustment, expressing a generally favorable assessment of the method and highlighting its ability to control the type \RNum{1} error rate. \citet{Schuler2020} demonstrated that the ATE estimate obtained through prognostic adjustment is efficient, i.e., it has the smallest possible asymptotic variance among a large group of estimators under the assumption of constant treatment effect. This implies that we achieve local semi-parametric efficiency by using prognostic adjustment, where "local" refers to the additional assumption of a constant treatment effect, which constrains the statistical model. Many of the previously mentioned techniques for covariate adjustment also exhibit local efficiency \cite{Tsiatis2008, MoorevanderLaan+2009, RosenblumvanderLaan+2010, Balzer2016}. For example, when employing a generalized linear model (GLM) as a working model for the outcome, \citet{RosenblumvanderLaan+2010} demonstrate local efficiency of a marginal effect estimand in the context of an RCT. In this scenario, "local" arises from the requirement that the working GLM model be correctly specified. Another example is TMLE, which is applicable to various types of data, and there exist numerous extensions of the method that leverage cross-validation to maximize efficiency \cite{vanderLaanGruber+2010, GrubervanderLaan+2010, Weixin+vdLaan2020, chen2023, RosenblumvanderLaan+2010, Balzer2016, Balzer2024}. In many cases, the estimator is locally semi-parametric efficient, with \cite{intro_modern_CI, vanderLaanRubin+2006, vanderLaan+2017} discussing the conditions necessary for TMLE to be non-parametrically efficient. \citet{Liao2023} broadened the concept of prognostic score adjustment to include already efficient estimators.

Prognostic adjustment for ATE estimation enhances the existing theory of locally semi-parametric efficient estimators. Not only does this method offer compelling theoretical properties, it is also easy to comprehend, making it an attractive addition to the toolbox of statistical analysis for RCTs. However, prognostic score adjustment for non-continuous outcomes for estimating marginal effect measures such as the odds ratio has not been established yet. Moreover, the pharmaceutical industry still commonly employs standard GLM models, which underscores the need for further exploration and the importance of investigating whether the integration of historical data within the context of GLMs can lead to enhanced study power. This exploration is particularly pertinent given the critical role of GLMs in pharmaceutical research and the potential for substantial gains in precision. 

In this paper, we explore the benefits of including an estimated prognostic score as an adjustment covariate in GLM based procedure suggested by \citet{RosenblumvanderLaan+2010}. We prove that in this case we get the same type of local efficiency under the assumption that the true treatment effect is additive on the link scale. We begin by introducing the setting and notation as well as the GLM plug-in procedure. We extend the GLM plug-in method by including Negative Binomial regression, broadening its utility for various data types. Additionally, we provide an analytical formula for the asymptotic variance that relies on only a few estimable population parameters, facilitating prospective sample size estimation. We address practical considerations and implications of our findings, offering insights for future clinical trial designs. We also present a simulation study that supports our theoretical findings and examines the benefits of the method in other data generating scenarios. Finally, we find decreases in the estimated variance through a analysis of clinical trial data supplied by Novo Nordisk A/S.

\section{Setting and notation}

We consider a two-armed trial involving $n$ participants, where the observational units, denoted $O_i = \left(W_i, A_i, Y_i\right)$, are independent and identically distributed for $i\in \{1, 2, 3, \ldots, n\}$, where $O$ are samples from the true but unknown distribution $\mathcal{P}$. Given the i.i.d. nature of the observations, we use the shorthand notation $O=\left (W,A,Y\right)$ without the index $i$ for a generic observation. Here, $Y$ represents a primary endpoint variable that can be continuous, binary, or ordinal. The vector $W$ comprises $p$ baseline covariates $W^1, W^2, W^3, \ldots, W^p$ collected at the initial participant visit. After gathering the necessary baseline information, participants are randomly assigned to their respective treatment groups. This assignment is denoted by the variable $A$, which is $1$ if the participant is randomized to the new intervention and $0$ if the participant is randomized to the control group. The probability of treatment is $P\left(A = a \right) = \pi_a$ where $0<\pi_a<1$ for $a\in \{0,1\}$. We do not make any parametric assumptions about the distribution of $Y$ given $(W, A)$ and use the empirical distribution for $W$. This implies that the statistical model $\mathcal{M}$ describing the range of potential observed data distributions is semi-parametric. 
Using this notation, the trial dataset can be represented as $\P_n = \left(\W, \A, \Y\right)\in \mathcal{W}^n\times \{0, 1\}^n \times \mathcal{Y}^n$, where $\mathcal{W}$ is the range of $W$, accommodating continuous, binary or categorical covariates and similar for the range of $Y$ denoted by $\mathcal{Y}$. We use $n_1$ and $n_0$ for the sizes of the treatment and control groups, respectively.

In our analysis, we adopt the causal inference framework introduced by \citet{Petersen+vdLaan2014} which is based on the Rubin causal model from \cite{Imbens2004} and \cite{Sekhon2008}. Here each participant has two potential outcomes: the outcome under the new treatment $Y(1)$ and with control treatment $Y(0)$. Let $\Psi_a=\E[Y(a)]$ be the population mean outcome under treatment $a$. We are interested in marginal effects, which are causal effects of the form 
\begin{align}
    \Psi = r\left(\Psi_1, \Psi_0\right),
\end{align}
where $r$ is continuously differentiable in $(\Psi_1, \Psi_0)$. We call them marginal effects since they only depend on the marginal mean outcomes. Examples of marginal effects are the difference in means $\Psi_1 - \Psi_0$ and the odds ratio $\frac{\Psi_1 / (1- \Psi_1)}{\Psi_0 / (1-\Psi_0)}$. {\color{black} For our purposes we also require that the treatment effect is decreasing in the control mean $\Psi_0$ and increasing in the treated mean $\Psi_1$, as would be intuitively expected.} Specifically,
\begin{align}\label{eq:asump_r}
    r_0'(\Psi_0,  \Psi_1) \coloneqq \frac{\partial r}{\partial \Psi_0}(\Psi_0,  \Psi_1) \leq 0 \quad \text{and} \quad r_1'(\Psi_0,  \Psi_1) \coloneqq \frac{\partial r}{\partial \Psi_1}(\Psi_0,  \Psi_1) \geq 0.
\end{align}
This sensible condition is satisfied by most definitions of
the treatment effect, including difference-in-means and odds ratio. Under randomized treatment we can identify any marginal effect from observed data because
\begin{align}\label{eq:estimand}
    \E[Y(a)] = \E\E[Y(a)| W] = \E\E[Y(a)| W, A = a] = \E\E[Y | W, A=a].
\end{align}
Throughout, we let $\mu(w, a)= \E[Y|W=w, A=a]$.

\section{Estimating marginal effects with generalized linear models}\label{sec:est}

GLMs, introduced by \citet{Nelder1972}, are powerful tools in the analysis of RCTs. They are e.g. used in the examination of binary data with Logistic regression or count data with Negative Binomial regression, both types of data which frequently represent clinically significant endpoints. In a GLM the outcome is modelled from a particular distribution from the exponential family including e.g. Normal, Binomial, Poisson, Gamma, Negative Binomial (with known stopping time parameter) or Inverse Normal distributions. The conditional mean $\mu(W, A)$ is modelled parametrically by $g^{-1} (\beta_0 +x \beta_x)$ where $g$ is the link function and $(1, x)$ is a row in the $n\times (1+q)$ design matrix $\X$ using $A$, $W$ and potentially treatment-covariate interaction effects. The linear predictor is then $\beta_0 \mathbb{1}_n +\X \beta_x$ where $\mathbb{1}_n$ is a vector of ones. The unknown parameters $\beta = (\beta_0, \beta_x)$ are typically estimated by maximum likelihood estimation (MLE).

The coefficients of the GLM are usually not interpretable as direct estimators of the marginal effects as is the case for the ANCOVA model (after demeaning the covariates) and the difference in means estimand \cite{HojbjerreFrandsen2024}. This is due to collapsibility problems for the measure of the causal effect. This intuitively means that including a baseline covariate in the linear part of the GLM model will change the treatment estimand, whenever that covariate is associated with the outcome \cite{Huitfeldt2019, Rhian2021}. For example, \citet{Rhian2021} describes why the odds ratio is a non-collapsible effect measure, since the difference between the marginal and conditional (on $w$) effect measure is not explained by sampling variation. Specifically, they show that even when there is no confounding due to randomization and when the outcome variable is binary, then the odds ratio will change when including a covariate that is associated with the outcome, no matter how large the sample size is.

However, GLMs can still be used to estimate marginal effects if care is taken. \citet{MoorevanderLaan+2009} applied TMLE to prove that in RCTs, certain estimators based on logistic regression are asymptotically unbiased even when the working model is misspecified. \citet{RosenblumvanderLaan+2010} generalized this to a larger class of GLMs. They describe a simple and extremely general plug-in method in the setting of an RCT, that we will briefly describe here. Specifically, they use the machinery of MLE for GLMs, but do not make any assumptions on the models being correctly specified, i.e., the working GLM model can be arbitrarily misspecified; in the model family, link function, and the linear predictor. However, it is necessary for the GLM to use canonical link functions and originate from frequently utilized families, which include Normal, Binomial, Poisson, Gamma, or Inverse Normal distributions which should not depend on the data, i.e. the family, link function, and linear predictor should be pre-specified. Additionally, the linear predictor must include an intercept and the treatment variable $A$. The design matrix can also include functions $f_j(W, A)$ for $j=1, 2, 3, \ldots, l$, but we will assume that all $f_j$'s are bounded on compact subsets of $\R^p \times \{0,1\}$. We also assume that the columns in the design matrix are linearly independent and that there exists a maximizer $\beta^*$ of the expected log-likelihood where the expectation is with respect to the true but unknown data distribution and that each component of $\beta^*$ has an absolute value smaller than a bound $b$. These are standard regularity conditions that guarantee the convergence of the parameter estimates of the GLMs due to the strict concavity of the expected log-likelihood for GLMs with canonical link function proved by \citet{R+vdL2009}. For the Gamma and Inverse Normal families of distributions, \citet{RosenblumvanderLaan+2010} impose an additional restriction to ensure that the canonical link is bounded. 

The strategy is to use a GLM as a working model to estimate $\mu(W, a)$ for each treatment $a\in \{0, 1\}$. First, find the MLE $\hat{\beta}$ and then estimate the conditional mean functions as $\hat{\mu}(w, a) = g^{-1}\left(\hat{\beta}_0 + x\hat{\beta}_x\right)$. To estimate a marginal effect, we extract counterfactual predictions from the GLM, first assuming everyone in the sample was actually treated ($a=1$), and then assuming the opposite ($a=0$). This specifically means that the outer expectation for the statistical estimand (right side of \eqref{eq:estimand}) can be estimated by the expectation under the empirical distribution of $W$, i.e. 
\begin{align}
    \hat{\Psi}_a = \frac{1}{n}\sum_{i=1}^n \hat{\mu}(W_i, a). 
\end{align}
The marginal effect estimate is then given by the plug-in estimator
\begin{align}\label{eq:estimator}
    \hat{\Psi} = r\left(\hat\Psi_1, \hat{\Psi}_0\right).
\end{align}

\paragraph{}
\citet{RosenblumvanderLaan+2010} show that the GLM based plug-in estimator is a regular and asymptotically linear estimator (RAL) assuming randomization of the treatment. RAL implies that the estimator is consistent for the marginal effect and has an asymptotically normal sampling distribution with variance described by the influence function (IF). The argument stems from the MLE for the GLM working model being an M-estimator \cite{Leonard2002}, implying that $\hat\beta$ is a RAL estimator for $\beta^*$. The GLM model that plugs into $\beta^*$ will be denoted by $\mu^*(w, a)=g^{-1}(\beta_0^*+x\beta_x^*)$. This is the large-sample probability limit of the GLM fit and can be viewed as a projection of the true distribution onto the GLM working model. This of course only implies that $\hat{\mu}(w, a)$ is a RAL estimator of $\mu^*(w, a)$ for a $w\in\mathcal{W}$ by the delta method. 

To obtain the consistency and asymptotic normality towards the causal estimand $\Psi_a$, \citet{RosenblumvanderLaan+2010} show that the scores solved by the M-estimator span the efficient influence function (EIF) for $\Psi_a$ under the assumed distribution $p_0(W, A, Y)=p_{0}(Y| W, A)\pi_A p_0(W)$, where $p_{0}(Y| W, A)$ is the density using the pre-specified GLM model and $p_0(W)$ is the empirical distribution of $W$. Consequently, the EIF estimating equation is solved using $W\indep A$ and $\Pe(A)=\pi_a$. This is then used to conclude that the procedure is a TMLE. Now, since the estimator is a TMLE they use Theorem 1 of \citet{vanderLaanRubin+2006} when incorporating five additional conditions that they verify within the context of an RCT. Theorem 1 of \citet{vanderLaanRubin+2006} allows us to conclude that the estimated counterfactual means, $\hat{\Psi}_a \pto \Psi_a$, are themselves consistent, asymptotically normal estimators of the counterfactual means for all of the GLMs considered. Applying the delta method does the rest, ensuring that our plug-in based estimator of the marginal effect is RAL (remain consistent and asymptotically normal), regardless of the type of misspecification. In addition \citet{RosenblumvanderLaan+2010} show that the estimator is locally efficient meaning that under correct GLM specification, this estimator achieves the efficiency bound for the model, assuming only an RCT setting. 

\citet{RosenblumvanderLaan+2010} exclude the Negative Binomial regression in their work. However, since the Negative Binomial distribution is often used in clinical trials to model count data to account for overdispersion, i.e. the variance exceeds the mean, we show in Appendix~\ref{app:NB} that the result does hold for the Negative Binomial regression as well. 

\paragraph{}
To derive the asymptotic variance of the marginal effect estimate via the delta method we need to use the asymptotic variance of the estimated counterfactual means. As per asymptotic theory, this is given by the variance of the IF for each counterfactual mean estimate, which for the GLM based estimator $\hat{\Psi}_a$ is 
\begin{align} \label{eq:IF_mean}
    \phi^*_a = \frac{1_a(A)}{\pi_a}(Y-\mu^*(W, a))+(\mu^*(W, a)-\Psi_a).
\end{align}
Given the consistency of $\hat{\Psi}_a$ it must follow that 
\begin{align} \label{eq:consistency}
    \E[\mu^*(W, a)]=\Psi_a=\E[Y(a)].
\end{align}
This implies that the plug-in estimate using the large sample limit $\mu^*(a, w)$, $\Psi_a^*$, would indeed equal $\Psi_a$. Using the delta method, the IF for the marginal effect estimator is 
\begin{align}\label{eq:IF}
    \phi^*_\Psi = r_0'(\Psi_1, \Psi_0)\phi^*_0 + r_1'(\Psi_1, \Psi_0)\phi^*_1,
\end{align}
implying that the asymptotic variance is $v_\infty^2 = \var(\phi^*_\Psi)$. Naturally without knowledge of the true data-generating process we do not know the probability limits $\Psi_a^*$, nor can we calculate the population variance. Therefore in practice we must use the sampling variance of the IF (same as the 2nd moment since IFs have mean-zero) where we plug in estimates of the unknown functions and quantities it depends on. This estimate is consistent for the asymptotic variance and is given by, 
\begin{align}\label{eq:asympvar}
    \hat{v}_\infty^2 = \frac{1}{n}\sum_{i=1}^n \left( r_0'(\hat\Psi_1, \hat\Psi_0)\cdot\hat\phi_{0}(A_i, W_i, Y_i)  + r_1'(\hat\Psi_1, \hat\Psi_0)\cdot\hat\phi_{1}(A_i, W_i, Y_i) \right)^2,
\end{align}
where $\hat\phi_{a}(W_i, A_i, Y_i) = \dfrac{1_a(A_i)}{\pi_a}(Y-\hat\mu(W_i, a))+(\hat\mu(W_i, a)-\hat\Psi_a)$. Confidence intervals (CI) and hypothesis testing can be obtained based on this estimate of the IF. 

{\color{black}
In the targeted maximum likelihood and double machine learning literature it is common to estimate asymptotic variances such as the one above using \textit{cross-fitting}. In general settings, $\hat\Psi$ can depend on complex machine learning estimates of various functions which breaks many asymptotic arguments. Cross-fitting resolves this. In cross-fitting, the functional estimates are generated from one (training) sample of data and those functions are evaluated on data from a different sample (estimation sample) to calculate asymptotic variances, etc. Usually this process is repeated round-robin, taking a small estimation fold of the data and leaving the rest for training each time \cite{dml}. In this way the cross-fit functions can be evaluated on all of the data without being trained on that self-same data.

Cross-fitting is not for $\sqrt n$-consistency in our setting since the estimated conditional means from a GLM typically lie in a Donsker class. Nonetheless, cross-fitting may improve finite-sample performance, and indeed we observe this in our simulation study in Section~\ref{sec:sim}. Through the rest of the paper we specify textually when estimates are cross-fit without evoking it in notation.
}

In practice we can account for the modeling process in the variance estimator by using a cross-validated variance estimate. The data is split into a validation and training data set respecting the randomization ratio. The GLM is fit using the training data set and is then used to estimate the IF for the observations in the validation data set. The sample variance of the cross-validated estimate of the IF is then used as a variance estimate for the marginal effect estimate. 
{\color{black}
The method we describe is not the only way to obtain inference for a marginal effect but it is simple and very general. As an alternative, one might also estimate the (jointly normal) asymptotic distribution of $(\hat\Psi_1, \hat\Psi_0)$ and analytically translate this to a more exact (non-normal) asymptotically-approximate distribution for $r(\hat\Psi_1, \hat\Psi_0)$. It is possible that such methods are necessary when the outcome is rare and the exact distribution of $r(\hat\Psi_1, \hat\Psi_0)$ is highly skewed. Since such methods would depend on the exact form of $r$ we do not consider them in this paper and leave them to future investigation. 
}

\subsection{Prospective power}\label{sec:power}

The expression for the asymptotic variance given in \eqref{eq:asympvar} is useful as a plug-in estimate for the variance once the data have been collected. However, it is not clear how it can be used for prospective power estimation or sample size calculations. We will now present a reduced form of \eqref{eq:asympvar} that only depends on a small number of population parameters that can either be prospectively estimated from historical data or guessed based on interpretable, realistic assumptions. This section thus represents a novel contribution, as it addresses the practical utility of the reduced form of \eqref{eq:asympvar} in the context of prospective power estimation and sample size calculations. The sample size calculation we propose here is essentially a generalization of the method of \citet{borm}, which has been used productively to design several completed trials \cite{rct1, rct2, rct3, rct4}.

First, let $\kappa_a^2=\E\left[(Y(a)-\mu^*(W, a))^2\right]$ be the expected mean-squared error of the best-possible GLM $\mu^*$ when predicting for $a$ and let $\sigma_a^2 = \var\left[Y(a)\right]$ be the marginal variance of the potential outcome under treatment $a$. $\mu^*$ can also be thought of as the probability limit or population-level version of $\hat\mu$. Furthermore, let $\tau = \cor\left[Y(0), Y(1)\right]$ be the population correlation coefficient between the two potential outcomes and $\eta = \cor\left[Y(0)-\mu^*(W, 0), Y(1)-\mu^*(W, 1)\right]$ be the population correlation coefficient between residuals of the best-possible GLM fit in opposite treatment arms. We will in \autoref{sec:power} discuss how to estimate these parameters. Starting from the above expression, the algebraic manipulation in Appendix~\ref{app:reducedformasympvar} gives, 
\begin{align}\label{eq:var_pop}
    v_\infty^2 = r_0'^{\, 2}\left( \frac{\pi_1}{\pi_0}\kappa_0^2 + \sigma_0^2\right) + r_1'^{\, 2}\left(\frac{\pi_0}{\pi_1}\kappa_1^2 + \sigma_1^2\right)-2|r_0'r_1'|\left(\tau \sigma_0 \sigma_1 - \eta \kappa_0\kappa_1\right),
\end{align}
where $r_a'$ denotes the derivative evaluated in $(\Psi_1^*, \Psi_0^*)$.

Now our reduced form for the asymptotic sampling variance depends only on the population parameters described in \autoref{sec:est}. If we want to prospectively estimate sampling variance for a sample size or power estimation, we need to estimate values of these quantities. In this section we will describe how to estimate these using historical data $\P_{\widetilde{n}} = \left(\widetilde{\W}, \widetilde{\Y}\right)\in \mathcal{W}^{\widetilde{n}} \times \mathcal{Y}^{\widetilde{n}}$ obtained from $\widetilde{n}$ control participants ($A=0$) that we assume to be taken from the same distribution that governs the trial.

While numerous population parameters mentioned above can be estimated using historical data or by making assumptions with clear interpretations, the cross-counterfactual correlations ($\tau$ and $\eta$) pose significant challenges. Estimating these parameters naturally involves treatment-arm data, which is unavailable to us in most cases during the planning phase of a trial. Furthermore, determining suitable values for these parameters without precise knowledge is also unclear. Nevertheless, progress can potentially be achieved by constraining these terms. By assuming $\tau\geq 0$ {\color{black}(positive correlation} between the two potential outcomes) and considering the worst-case scenario $\eta=1$, we can conservatively limit the overall asymptotic variance, 
\begin{align}\label{eq:var_bound}
\begin{split}
    v_\infty^2 &\leq r_0'^{\, 2}\left( \frac{\pi_1}{\pi_0}\kappa_0^2 + \sigma_0^2\right) + r_1'^{\, 2}\left(\frac{\pi_0}{\pi_1}\kappa_1^2 + \sigma_1^2\right) + 2|r_0'r_1'|\kappa_0\kappa_1\\
    &= r_0'^{\, 2}\sigma_0^2+ r_1'^{\, 2}\sigma_1^2+ \pi_0\pi_1\left(\frac{|r_0'|\kappa_0}{\pi_0} + \frac{|r_1'|\kappa_1}{\pi_1} \right)^2\\
    &\equiv v_\infty^{\uparrow 2}.
\end{split}
\end{align}
Since the design $\pi_a$ is fixed and the marginal outcome variances $\sigma_a$ are fixed the bound implies that the root mean squared generalization errors ($\kappa_a$) of the limiting GLM are the primary factor influencing the reduction of large sample variance and consequently increasing statistical power. This means that we should include as many baseline covariates with high predictive power as practically feasible in the GLMs. It is also desirable to include pre-specified interactions and transformations thereof. 

{\color{black}
To help understand when $\tau \ge 0$, let $Z$ denote a (possibly unobserved) set of covariates such that, conditional on $Z$, the remaining idiosyncratic components
of $Y(0)$ and $Y(1)$ are uncorrelated, i.e. $Y(1) \perp Y(0) | Z$\footnote{This is automatic if 
$Y(a)=\mathbb E\{Y(a)\mid Z\}+\varepsilon_a$ with $\mathbb E(\varepsilon_a\mid Z)=0$
and $\mathrm{Cov}(\varepsilon_0,\varepsilon_1\mid Z)=0$.}. Following arguments by \citet{Schuler+2022+151+171}, a sufficient condition for $\tau\ge 0$ is
\[
\var\big(\mathbb E[Y(1)-Y(0)\mid Z]\big) \le\ \var\!\big(\mathbb E[Y(a)\mid Z]\big),
\]
for either $a=0$ or $a=1$.
Thus, $\tau\ge 0$ holds whenever treatment-effect heterogeneity across \textit{any set of definable strata} is not more
variable than the heterogeneity (across strata) in expected outcomes. If one cannot imagine a way to stratify a trial population such that the additive treatment effect would be more heterogeneous across those strata than the expected control or treated outcomes themselves would be, then $\tau \ge 0$ is a sensible assumption. Roughly speaking, $\tau \ge 0$ fails only if treatment heterogeneity is particularly extreme.
}

The conservative bound $v_\infty^{\uparrow 2}$ enables us to determine the power prospectively by estimating values for $r_a'$, $\sigma_a^2$ and $\kappa_a^2$ from the historical control data $\left(\widetilde{\W}, \widetilde{\Y}\right)$ or from expert knowledge. To get a value for $r_a'$ we need $\Psi_1$ and $\Psi_0$. For $\Psi_0$ a natural estimator is the sample mean, $\hat\Psi_0=\frac{1}{\widetilde{n}}\sum \widetilde{Y}_i$. To determine $\Psi_1$ we only need to consider the target effect size {\color{black}(minimum clinically important difference)} $\Psi$ and solve $r(\hat\Psi_1, \hat\Psi_0)=\Psi$. Plugging these values into the analytically derived derivative gives the needed value. 
To estimate the best case treatment specific average MSEs $\kappa_a^2$ we need to approximate ${\mu}^*(w,a)$. This is done by fitting the same GLM to the historical data as specified for the current trial. The resulting fit is then $\hat{\mu}(w, 0)$ since all of the of the historical participants have $A=0$. Thus the plug-in for $\kappa_0^2$ is just the empirical MSE of this fit,
\begin{align}
\begin{split}
    \hat\kappa_0^2=\frac{1}{\widetilde{n}}\sum \left(\widetilde{Y}_i-\hat{\mu}(\widetilde{W}_i, 0) \right)^2.
\end{split}
\end{align}
{\color{black}This MSE should be assessed out-of-sample on the historical data, or in a cross-validated fashion, so that the fit $\hat\mu$ is not being evaluated on the same data points $\widetilde W_i$ on which it was trained.}

The historical data does not contain information about the treatment arm, so we cannot use the same procedure to estimate $\kappa_1^2$. A reasonable and interpretable assumption would be $\hat\kappa_0^2=\hat\kappa_1^2$. One may conduct a sensitivity analysis by inflating this value to assess the robustness of the power estimation. 

The marginal outcome variance for the control group is trivially estimated by, 
\begin{align}
\begin{split}
    \hat\sigma_0^2=\frac{1}{\widetilde{n}}\sum \left(\widetilde{Y}_i-\hat{\Psi}_0 \right)^2.
\end{split}
\end{align}
{\color{black}
There are various ways to estimate the treatment-arm variance, depending on modeling assumptions. For example, in the case of binary outcome data we have $\sigma_1^2=\Psi_1(1-\Psi_1)$, so plugging in the counterfactual mean $\hat\Psi_1$ based on the target effect size would yield an appropriate estimate. On the other hand, when the outcome is continuous, it is often sensible to assume $\sigma_1 \approx \sigma_0$. Thus a sensible estimate is $\hat\sigma_1 = \hat\sigma_0$.
}

\paragraph{}
The procedure for doing a prospective power calculation can be summarised as: 
\begin{enumerate}
    \item Set design parameters $\pi_a$ and $\Psi$.
    \item Collect the historical data, that are representative of the new study in design, population and the type of data being collected. 
    \item Use these to estimate  $r_a'$, $\sigma_a^2$ and $\kappa_a^2$ with the procedure outlined above or sensible alternative based on context.
    \item Plug the estimates into \eqref{eq:var_bound} to obtain an estimate of $v_\infty^{\uparrow 2}$. 
    \item Now let $\Delta$ be the marginal treatment effect in case of no treatment difference, e.g. $0$ for ATE and $1$ for the rate ratio. Calculate the lower bound for the power by comparing $\He_0: \hat{\Psi} \sim F_0 = \N(\Delta ,v_\infty^{\uparrow 2} / n)$ to $\He_1: \hat{\Psi} \sim F_1 = \N(\Psi,v_\infty^{\uparrow 2} / n)$. Specifically, the power can be determined from $F_1\left(F_0^{-1}(1-\alpha/2)\right)=\gamma$, where $\alpha$ is the significance level and $1-\gamma$ is the power. A required sample size can be determined by increasing $n$ until power exceeds the desired amount.
\end{enumerate}


{\color{black}The only assumptions required for this sample size calculation to be valid (exceed design power) are that $\hat r'_a \ge r'_a$, $\hat\sigma_a^2 \ge \sigma_a^2$, and $\hat\kappa_a^2 \ge \kappa_a^2$, i.e., the population parameters are estimated conservatively, and that $\tau \ge 0$. The latter condition holds when the (additive) treatment effect function is not more variable than the conditional mean function of the outcome of either treatment group, e.g., when the effect is \textit{close enough} to constant. Practically speaking, one can ask: is there a set of hypothetical covariates for which the treatment effect varies more across covariate strata than the predicted outcome does across those same strata?
}

\section{Leveraging a prognostic model}\label{sec:prog_GLM}

{\color{black}
We now examine how historical data can be used to increase efficiency without incurring bias.

Suppose that at the design phase of the study we have access to a historical dataset of $\widetilde{n}$ patients containing covariate and outcome data, which we denote $\left(\widetilde{\W}, \widetilde{\Y}\right)$. 
Whereas we previously used the symbol $P$ to refer to the distribution of variables $(W, A, Y)$ in the trial population, we now change the definition to understand $P$ as the distribution of $(W, A, Y, D)$ in some broader population where $D=1$ indicates membership in the trial population and $D=0$ the population from which the historical dataset is drawn from. In this context, what we previously wrote as $P(W, A, Y)$ is now $P(W, A, Y|D=1)$ and our trial dataset is sample of size $n$ from this distribution. 
For this paper we assume that we only have access to historical \textit{control} data, i.e. $P(A=0|D=0)=1$. 
Given this knowledge, we consider the historical data to be a draw of size $\tilde n$ from the distribution $P(Y,W|D=0)$. 

Similarly, we redefine $\mu$ from $\mu(W,A) = \E[Y|W,A]$ to $\mu(W,A,D) = \E[Y|W,A,D]$. The former function is now $\mu(W,A,1)$. With this notation, we write the trial-specific treatment-specific means as $\Psi_a = \E[\mu(W,a,1)]$ and the trial-specific marginal effect $\Psi = r\left(\Psi_{1}, \Psi_{0}\right)$.
This encodes the fact that we are interested in the treatment effect \textit{in the trial population only}. The trial-specific effect still identifies its causal counterpart since treatment is randomized in the trial, i.e., $P(A|W,D=1) = 1/2$.
}

We define the \textit{prognostic score} as the expected observed outcome under the condition that the observation originates from the historical control data (where $A=0$): 
\begin{align}
    \begin{split}
        \rho\left(W\right) \coloneqq \E[Y \,|\, W,A=0, D=0].
    \end{split}
\end{align}
Let $\mathcal{F}$ represent a machine learning algorithm that links historical outcomes and covariates to a function $f:\mathcal{W} \to \mathcal{Y}$. The machine learning algorithm may include model selection and parameter tuning. The fitted function then estimates the conditional mean $\E[Y \, | \, W, D=0]$. Thus, we can use $\hat\rho = \mathcal{F}\left(\widetilde{\W}, \widetilde{\Y}\right)$ learned from historical data (with $\widetilde{n}$ observations) to estimate the prognostic score of a participant with covariates $w$. 

{\color{black}
If the control-arm RCT data and the historical data are sampled from the same distribution, i.e. $P(W,0,Y|D=1) = P(W,0,Y|D=0)$, then the prognostic score is also the control-arm conditional mean function in the trial, i.e. $\rho(w) \coloneq \mu(w,0,0) = \mu(w,0,1)$. This motivates us to treat the estimated prognostic score (as predicted on the trial observations) as a covariate in the trial analysis, since we expect it to be predictive of the outcome in the control arm at a minimum.
}

Specifically, our proposal is to use the GLM based marginal effect estimator as in \autoref{sec:est}, now including the prognostic score on the link scale, i.e., we include $g(\hat\rho(W))$ as an additional covariate in the design matrix. {\color{black}Including the prognostic score as a predictor on the link scale makes it so that the GLM simply predicts the prognostic score if its coefficient is set to one and all other coefficients to set to zero. In other words, it makes it easy for the GLM to simply ``pass through'' the prognostic score.} Again, we allow for any pre-specified adjustments of additional baseline covariates and interactions since the prognostic score enters like any other covariate.

Intuitively, the effectiveness of prognostic score adjustment in GLM models can be attributed to its ability to capture complex relationships between the outcome $Y$ and the covariates $W$. As discussed in \citet{HojbjerreFrandsen2024}, by constructing a prognostic score that explains a significant portion of the variation in $Y$, we can substantially reduce residual variance and enhance statistical power for linear models. The same holds for GLM models since such a prognostic score would lead to a reduction in $\kappa_a$ leading to smaller CIs and p-values. This is particularly advantageous when the relationship between $Y$ and $W$ can not be modelled by an GLM model or involves interactions among covariates. In such cases, traditional GLM adjustment methods may fail to capture these intricate patterns, leading to suboptimal adjustments. By leveraging machine learning techniques to develop the prognostic score on historical data, we can detect and model non-linear and subgroup effects, thereby capturing nuanced relationships that GLM models might miss. This approach not only mitigates the risk of overfitting by performing a sort of variable selection on historical data but also ensures that the adjustment is robust to the complexities inherent in the data. The prognostic score thus offers an opportunity to exploit the information present in all covariates without necessarily having to include a large number of main terms, splines, interactions, etc., as suggested by \citet{tukey1993tightening}. 

An illustrative example of this can be found in \citet{HojbjerreFrandsen2024}, where the authors demonstrate the superiority of prognostic score adjustment in scenarios where the relationship between $Y$ and $W$ cannot be adequately modeled using standard GLM methods. The success of the method obviously depends on the quality of the prognostic score; if the prognostic score is noisy (e.g., trained on a small dataset or a very different population), there will be no benefit. However, there will also be very little or no harm - even in very small samples. The only cost is one added degree of freedom in the GLM, and asymptotically, there is no cost at all.

\subsection{Efficiency}

Since the estimated prognostic score $\hat\rho$ is dependent on the historical data size $\widetilde{n}$, we can regard this a sequence of random functions $\{\hat\rho_{\tilde n}\}_{\widetilde{n}\in\mathbb{N}}$. 
Such a sequence is \textit{uniformly bounded} if there exists $K>0$ such that $P(\,|\hat\rho_{\tilde n}\, (W)| \le K)=1$ for all $\tilde{n}\in\mathbb{N}$. 
Secondly, we say, for $k>0$, that $\{\hat\rho_{\tilde n} \,(W)\}_{{\widetilde{n}}\in\mathbb{N}}$ \textit{converges in} $L^k$ towards the stochastic variable $\rho(W)$ if 
\begin{align}
    \lim_{{\widetilde{n}}\to \infty}\E_{W, \, \hat\rho_{\tilde n}}\left[\,|\hat\rho_{\tilde n}\,(W)-\rho(W)|^k\right]=0, 
\end{align}
denoted as $\hat\rho_{\tilde n}\,(W)\overset{L^k}{\longrightarrow} \rho(W)$. For convenience, we will sometimes refer to the sequence $\{\hat\rho_{\tilde n}\}_{{\widetilde{n}}\in\mathbb{N}}$ as just $\hat\rho_{\tilde n}$ or $\hat\rho$.

\citet{Schuler2020} show that using an ANCOVA model adjusting for an estimated prognostic score that converges in the $\L_2$ sense to the true conditional counterfactual mean of the control group in the current trial ($D=1$) under a homogeneous treatment effect will lead to an efficient estimate of the ATE {\color{black}when treatment is randomized}. In this Section we will show that including a prognostic score in a GLM will lead to the same type of local efficiency of the marginal effects estimator. We note here that it is a type of local efficiency due to the additional assumption of homogenenous treatment effect.

\begin{restatable}{theorem}{restatedtheorem}\label{thm:main}
Let $\P_n$ and $\P_{\widetilde{n}}$ be the empirical distributions of $n$ and $\tilde n$ draws from $P(W,A,Y|D=1)$ and $P(W,Y|D=0)$, respectively. 
Presume that the number of participants $n$ in the current trial increases such that $n=\mathcal{O}\left(\widetilde{n}\right)$.  Furthermore, let $\hat\rho_{\tilde n}$ be a uniformly bounded random function learned from the historical data $\P_{\widetilde{n}}$ obtained from $\widetilde{n}$.

Assume that the true treatment effect is additive on the link scale, i.e. $g\left(\mu(W, 1)\right) = \zeta + g\left(\mu(W, 0)\right)$ and that $\mathcal{W}$ and $\mathcal{Y}$ are compact. If the prognostic score converges to the control-arm conditional mean, i.e., $\left|\hat\rho_{\tilde n}\,(W) -\mu(W, 0)]\right|\overset{L^2}{\longrightarrow} 0$ as $\widetilde{n}\to\infty$, 
then the estimator obtained from the plug-in GLM procedure in \eqref{eq:estimator} suggested by \citet{RosenblumvanderLaan+2010} that uses $g(\hat\rho_{\tilde n}\,(W))$ as an additional covariate is consistent and efficient in the sense that it has the lowest possible asymptotic variance among all RAL estimators with access to $W$.
\end{restatable}

The proof and supporting lemmas for this result can be found in Appendix~\ref{app:proof}. \autoref{thm:main} means that using the plug-in GLM method combined with prognostic score adjustment will result in the expected smallest valid CIs and p-values that are statistically possible to obtain with access to $W$ as long as there is an additive effect and the prognostic model is well-specified. 

{\color{black}
The condition that the treatment effect is additive on the link scale is equivalent to there being a constant conditional treatment effect, e.g. constant conditional odds ratio. This assumption can be eliminated, resulting in general efficiency, if we additionally have a prognostic score for $a=1$, i.e. $\rho_1(w) = \E[Y|W,D=1,D=0]$ such that $\left|\hat\rho_{\tilde n, 1}\,(W) -\mu(W, 1)]\right|\overset{L^2}{\longrightarrow} 0$ as $\widetilde{n}\to\infty$ that we include on the link scale along with interaction terms between treatment and covariates. In that case we can show using a similar argument to that of \autoref{theorem:true_prog} to conclude that $\mu^*(W, A)= \mu(W, A)$ and thus follow the same argument as for \autoref{thm:main} in \autoref{app:thm_main}. Since treatment arm data are rarely available in any meaningful quantity prior to the trial we restrict ourselves to having only control-arm historical data and prognostic scores in this paper.
}

We have demonstrated that adding a prognostic score under specific conditions leads to asymptotic efficiency in our estimates. This is a best case scenario. What happens when we relax the assumption that the prognostic score converges to the true conditional mean functions? In such cases, will adding the additional covariate always yield an efficiency gain? The answer is no, but in what follows we formalize under what scenarios it is impossible to hurt performance, at least asymptotically.

\citet{Schuler2020} showed that when the treatment probabilities are equal ($\pi_1 = \pi_0$), or when there is a constant treatment effect, including covariates in the \textit{linear} regression model never decreases the asymptotic efficiency of the ATE estimate compared to the unadjusted estimator. However, \citet{FREEDMAN2008180} show that for the logistic regression with $\pi_1 \neq \pi_0$ adjusting for additional covariates may hurt the asymptotic variance. 

{\color{black}
In our setting we formalize this problem by considering two nested regression models and their respective associated point estimates and variances when used to estimate marginal effects. One example of such a nesting is two regression models where the larger one has one additional covariate relative to the smaller. In our setting one may imagine that the additional covariate is the prognostic score, but the theorem is agnostic to this.
}

\begin{restatable}{theorem}{nested}\label{thm:nested}
Let $\widetilde\Psi$ and $\overline\Psi$ be estimators with IFs $\widetilde\phi$ and $\overline\phi$ indexed by functions $\widetilde\mu$ and $\overline\mu$ respectively as described by the construction in \autoref{eq:IF_gen}. If $\widetilde\mu$ and $\overline\mu$ are minimizers of population mean-squared error in nested models $\overline{\mathcal M} \subset \widetilde{\mathcal M}$ with $\widetilde{\mathcal M}$ a closed subspace of $\mathcal L_2$, then the asymptotic variance of $\widetilde\Psi$ is less than that of $\overline\Psi$.
\end{restatable}

The proof is in Appendix~\ref{app:ATE_nested}. This result indicates that if the regression coefficients are estimated by minimizing the mean-squared error, then incorporating the prognostic score does not increase the asymptotic variance. This aligns with the findings of \citet{Schuler2020} in the context of linear models, where maximizing the likelihood implies minimizing mean-squared error. Conversely, when we employ maximum likelihood (rather than least squares) to fit the GLM, we cannot guarantee that an additional covariate will not lead to an increase in asymptotic variance in certain scenarios. Nonetheless, adding even a non-informative covariate did not meaningfully increase empirical variance in any of our simulations. In practice, we suspect it is challenging to identify a scenario in which including a predictive covariate hurts the asymptotic variance, particularly when using the GLM plug-in approach. 

\subsection{Practical implementation of prognostic score adjustment for GLMs}\label{sec:rec}
In this subsection, we will discuss the practical implementations of the prognostic score method specifically tailored for GLMs. The recommendations provided here are closely aligned with those presented by \citet{HojbjerreFrandsen2024}. However, for the sake of completeness and to emphasize the practical aspects unique to GLMs, we restate them here. For a more in-depth understanding, readers are encouraged to refer to \citet{HojbjerreFrandsen2024} as well as the guidelines \cite{Handbook} by EMA on prognostic score adjustment for linear models. This discussion aims to provide clear guidance on how to effectively apply the prognostic score method in the context of GLMs, ensuring that practitioners can leverage its benefits in their analyses.

\noindent{}\textbf{Prospective power using a prognostic model}\\
The same powering strategy as in \autoref{sec:power} can be applied when using a prognostic model. The only difference or simplification comes with estimating $\kappa_0^2$. Since $\hat\rho_{\tilde n}(W)$ approximates the true conditional mean $\mu(W, 0)$, in the limit the GLM should simply pass this prediction through. Thus, $\mu^*(W, 0)\approx \hat\rho_{\tilde n}(W)$ and we might estimate 
\begin{align}
\begin{split}
    \hat\kappa_0^2=\frac{1}{n'}\sum \left(\widetilde{Y}_i-\hat\rho_{\tilde n}(\widetilde{W}_i) \right)^2,
\end{split}
\end{align}
where the historical data have been separated into a data set for training the prognostic model and one for testing. This new $\hat\kappa_0^2$ is the test data set MSE of the prognostic model. Since the GLM that includes the prognostic model as an additional covariate has more freedom to fit the data, this estimate will always be larger than the alternative using $\hat{\mu}(W, 0)$. Thus, this estimate will be conservative under the assumption that the conditional mean does not change from the historical population to the current trial population. 

\noindent{}\textbf{Curation of historical data}\\

While the theoretical optimality of the prognostic score method is well-established, its practical application requires careful consideration. To ensure the method's effectiveness, it is crucial to have a substantial amount of high-quality historical data that is independent of the new study. Ideally, this historical data should be representative of the new study population and the type of data being collected, adhering to Pocock's criteria for prognostic model adequacy \cite{Pocock1976}.

{\color{black}
While having historical data that is from a similar population as the trial is beneficial, it is not strictly necessary for improved efficiency. As long as the prognostic model is predictive to any extent in the trial population, there will be some large-sample power gain over an unadjusted analysis. However, the power gained is less if the model is less predictive. 

How predictive the model turns out to be is a function of both the amount of historical data and the population(s) they are drawn from. Large historical sample sizes drive down the variance (of the estimated prognostic function) and can thus increase out-of-sample predictive performance. But if the additional data comes from the wrong population it increases the bias (of the prognostic model, not of the trial effect estimate) and thus \textit{decreases} predictive performance. Therefore, when the historical data is not from the correct population there is a bias-variance tradeoff in deciding how much of it to use. 

Since adding data to predictive model training usually has diminishing returns after a large enough sample size, we recommend being selective only if the amount of historical data is large (e.g. tens of thousands of observations). Historical data can be matched to the proposed trial by comparing inclusion/exclusion criteria, looking at covariate distributions, and comparing metadata such as time and site. Data-driven methods like test-then-pool or survey weighting (during training) can also be used if some trial pilot data or trusted matched historical data is available. On the other hand, if there is little historical data, (e.g. low hundreds), then being selective could hurt the predictive ability of the model more than whatever bias is incurred by leaving the out-of-population data in. 

The impact on bias and variance is measurable to the extent that one can compare cross-validation error of the prognostic model trained with and without the data considered for removal, within and outside of these data. If performance increases substantially when adding the less well-matched data compared to that on the smaller sample of historical data that are strongly matched to the trial, then that could be evidence of a favorable bias-variance tradeoff, and vice-versa. This is also contingent on having a well-matched subset of historical data that one is sure is free from substantial \textit{unmeasured} differences from the trial population that would also affect patient outcomes. 

However, if the subset of well-matched data is too small, the error in measuring the performance itself may be too large to understand the impact one way or the other. 
Sufficient historical data are needed to partition them into training and test datasets (or run cross-validation), which are used to estimate population and prognostic model performance parameters for sample size determination. 
}

Lastly, the historical and current data should be in the same format to facilitate data collection and structuring for prognostic model building. Allocating sufficient resources early in the trial process to organize the historical data into a comprehensive subject-level dataset is essential. Large integrated databases, potentially involving data sharing between pharmaceutical companies, may be necessary to achieve this. {\color{black}Care must be taken during data harmonization to preserve the semantic meaning of different variables if reasoning about population differences depends on that meaning.}

{\color{black}
Ultimately, while we can provide some guidance on how to select historical data, there is no substitute for good judgment and expertise. Prognostic adjustment, and especially sample size calculation based on prognostic adjustment, should be undertaken only under the supervision of trained statisticians familiar with both machine learning and trial analyses and domain experts who understand the data collection and clinical course of care.
}

\noindent{}\textbf{Sensitivity of sample size estimation}\\
{\color{black}
If prognostic score adjustment is only used as the analysis but not accounted for in sample size calculation, the worst that can happen with large enough samples is that the study is properly powered, but not overpowered. However, if the trial is powered based on the assumption of good prognostic model performance and the model under-performs, then the trial may be \textit{under}powered. Therefore the guidance in the preceding section should be carefully adhered to if \textit{powering} a study with prognostic score adjustment. 

Moreover, it may be sensible to calculate the required population values $(\sigma, \kappa)$ based on a smaller subset of well-matched data, even if the prognostic model is trained on larger datasets. Here there is also a bias-variance tradeoff: if the subset of ``matched'' data is too small, e.g. for estimating the variance of the out-of-sample $\kappa$, then small-sample variance could be worse than the bias incurred from using a less-curated dataset to estimate these parameters. Thankfully, however, since the population parameters are averages, their estimates have $\sqrt{n}$-convergence. Therefore in most settings even low hundreds of samples should be sufficient to drive the variance down.
}

Sensitivity analyses for the sample size estimation can also be conducted by inflation of $\sigma_a^2$ and $\kappa_a^2$. It should however be noted that the sample size is already conservatively estimated even when the model is misspecified since it is based on the upper bound of the asymptotic variance in \eqref{eq:var_bound}. Similar to the guidelines in \cite{Handbook} a sensitivity analysis using an inflation parameter above 1 should be considered if one of the following applies: 1) if only one test data set similar to the new trial is available 2) if there are changes in the standard of care 3) if the prognostic model includes covariates thought to be highly predictive of heterogeneity in treatment. If the prognostic model incorporates such covariates, the expected mean squared error $\kappa_a^2$ might be larger (or smaller) for the new treatment group compared to the control group. In that case one might consider inflating $\kappa_1^2$ but not $\kappa_0^2$. Considerations on parameter choices for the sample size determination must be prespecified in the SAP.

All in all, use of a non-representative external dataset can lead to an underpowered study, which is unethical as it unnecessarily exposes participants to potential harm. This risk falls on the sponsor. The responsibility for assessing the representativeness of the external dataset lies with the sponsor, who must also convince the ethics committee that the study is not underpowered due to the use of linear prognostic score adjustment. If a suitable external dataset is not available, it is not advisable to \textit{power} the trial for prognostic adjustment, but it is still sensible to use prognostic adjustment to \textit{analyze} the trial.

\noindent{}\textbf{Prognostic score construction and adjustment}\\
We recommend using an adaptive, ensemble machine learning method (``super-learning'') \cite{superlearner} for constructing the prognostic score. In this approach, flexible models (e.g., splines, regression trees) can be combined with simple linear models, offering robustness and flexibility. Including simpler models in the ensemble helps to avoid overfitting. {\color{black}Guidance on appropriate use of ensemble learning may be found in \citet{phillips2023practical}.}

It is also advisable to include missingness indicators as inputs for the prognostic model and to handle covariate missingness as in any other RCT analysis. If there is different patterns of missingness between the historical data and new data, the inclusion of missingness indicators as an addition to imputation of the missing data might increase the accuracy of the prognostic model in prediction of outcomes in the new trial.

Resources should be allocated early to build and validate the model following good machine learning practices \cite{FDA_ml}. Model selection and tuning parameters must be prespecified in the SAP, and the prognostic model should be finalized before unblinding. Data scientists providing the prognostic scores should be blinded to the randomization code.

Other adjustment covariates and the choice of variance estimator should be specified in the SAP, in line with regulatory recommendations \cite{FDA_cov}. Linear adjustment for a prognostic score can be combined with multiple imputation using the estimand framework \cite{ICH_est}, following Rubin's rules. This should also be pre-specified in the SAP. The sponsor should conduct sensitivity analyses for recruitment bias, complete losses to follow-up, and treatment compliance, as in any other trial.

\section{Simulation study}\label{sec:sim}

{\color{black}
The primary aim of these simulations is to assess the extent to which GLM prognostic score adjustment increases efficiency (smaller standard errors) relative to standard covariate adjustment or unadjusted estimates. Second, we seek to assess whether GLM prognostic score adjustment preserves coverage, which assesses both the bias and the accuracy of standard errors. Lastly, we investigate whether our proposed power calculation is effective at achieving design power when the historical data are indeed drawn from the same population as the trial data, and to what extent it fails otherwise.

In these studies we simulate a count outcome and presume the estimand of interest is the rate ratio $\Psi_1/\Psi_0$. 
We set up data-generating processes for both the trial and historical populations, draw historical and trial data, and use these data to generate estimates of the rate ratio according to unadjusted, covariate-adjusted, and prognostically-adjusted estimators. We then compare the estimates to ground truth to assess various properties. We test across multiple data-generating processes to see how robust the method is against conditions designed to compromise it, especially changes between the historical and trial populations. Details are in the sections below and code for these simulations can be found \texttt{\href{https://github.com/NNEHFD/GLM_prog}{here}}.
}

\subsection{Data generation}\label{sec:setup}
{\color{black}
The simulation study is based on the structural causal model in \eqref{eq:DPG}. We will first present the data-generating process and then explain how we use it to address our simulation aims. The scaffold of the data-generating process is
\begin{align}\label{eq:DPG} 
\begin{split}
    U &\sim \mathcal N(u_d, 1) \\
    W_1 &\sim \mathcal N(w_d, 1) \\
    W_{2, \dots 5} &\sim \mathcal N(0,1) \\
    A &\sim 
    \begin{cases}
    \text{Bern}(0.5) & d = 1\\
    0 & d = 0
    \end{cases} \\
    Y(a) &\sim \text{Pois}(m(U, W, a)) \\
    Y &= Y(A).
\end{split}
\end{align}
Here $d$ indicates whether we are simulating from the trial population ($d=1$) or historical population ($d=0$). The mechanism for $A$ ensures randomization in the trial and $A=0$ deterministically in the historical data. The mean $w_d$ of the first covariate, the mean $u_d$ of the unmeasured covariate $U$, and the nature of the $(U,W,A)$-conditional mean $m$ are parameters that we modify to create different scenarios of interest, as discussed below.

We use the notation $m(U,W,A) = E[Y|U,W,A]$ to distinguish this function from $\mu(W,A,D) = E[Y|W,A,D]$. The latter conditional mean is what matters. Our prognostic score will attempt to model the historical control-arm predictive mean $\mu(W, A=0,D=0)$: if this is a good proxy for the trial means $\mu(W, A=0,D=1)$ and $\mu(W, A=1,D=1)$ then we expect to see a nice performance boost adjusting for the estimated prognostic score. The function $m$ (in conjunction with the unobserved variable $U$) is just a tool that helps us generate effect heterogeneity and structural changes between the historical and trial distributions. 

In our simulations, the dependence of the $W$-conditional mean $\mu(W,a,d)$ on the population $d$ is through the distribution $U|D$ and the function $m$. Although $m$ is fixed, if $U|D=1$ is different than $U|D=0$, the mean of $Y(a)$ conditional only on $W,D$ changes if $U|D$ changes. Therefore modifying the \textit{distribution} of the unobserved variable from the historical to trial populations induces changes in the \textit{functional relationships} between the outcome and other observed variables in the two populations. This reflects the fact that ``probabilistic'' relationships actually arise from structural, mechanistic relationships involving unmeasured variables. Moving between populations, the mechanisms of disease, laws of physics, etc. do not change, but the input variables might. Thus, as far as distributions over observables in a structural model are concerned, there is no difference between changing the distribution of exogenous variables and changing the structural equations themselves. 

In all of our simulations, we set the control-arm conditional mean to
$$
m(U, W, 0) = 0.1 + 2h(W_1+1) + W_2^2 + h(W_1 W_4) + |U|h(W_3+2)
$$
where $h$ is the hinge function at $h(z) = (z)\times 1(z > 0)$. We use the hinge function as a convenience to obtain partly linear relationships without suffering negative conditional means (our outcome is a positive integer so the mean must be positive). Since the covariates are standard normals (mainly in the range of $[-2,2]$), setting $t$ on the range of 2 in $h(W+t)$ makes such a term predominantly linear. Smaller values of $t$ induce more nonlinearity.

This choice of $m$ has some moderate nonlinearities and an interaction. The percent of variance of $Y$ explainable by $W$ according to this true conditional mean is about $0.5$, while the amount that is explainable by linear functions of $W$ is only about $0.3$ (based on running a regression in a very large sample). This was by design, to reflect a linear $R^2$ within the range one would expect for a trial and to leave some room for improvement by modeling nonlinearities.

Since $U$ enters $m$ in its own term and is independent of $W$, we can marginalize it out to see
$$
\mu(W,0,d) =  0.1 + 2h(W_1+1) + W_2^2 + h(W_1 W_4) + \E[|U|]h(W_3+2)
$$
where $\E[|U|]$ depends on $u_d$, the mean of $U$. Therefore we see that shifting the unobserved covariate makes the value of the coefficient on the $W_3$ term different between the trial and historical populations.

In the treatment arm, we set
$$
m(u, w, 1) = g^{-1}(\zeta + g(m(u, w, 0) + 2\eta h(W_4))
$$
where $g(z) = \log(z)$ is our link function and $\zeta$ controls the size of the treatment effect. $\eta$ is a binary variable that turns heterogeneity of effect on or off. The heterogeneity is on the scale of the other terms in the conditional mean function so it constitutes a moderate amount of heterogeneity.
When $\eta=0$, (link-scale additive effect) the rate ratio is given exactly by $e^\zeta$.

Using this scaffold, we define three trial scenarios shown in \autoref{tbl:trial-scenarios} and five modifications for the historical DGP for each shown in \autoref{tbl:hist-scenarios}. This is described in further details below.

\subsubsection{Trial Scenarios}

In our first trial scenario we encode a null effect in order to be able to compare the type \RNum{1} error rates of each estimator. For this scenario, in the trial DGP we encode the covariate means $u_1 = w_1 = 0$, $\eta = 0$ (no heterogeneity), and $\zeta=0$ (null effect; rate ratio equals one). 

In our second trial scenario we encode a constant (link-additive) effect, since our theory predicts that efficiency gains may be higher under constant effects and sample size calculations more conservative. For this scenario, in the trial DGP we encode the covariate means $u_1 = w_1 = 0$, $\eta = 0$ (no heterogeneity), and $\zeta=0.2$, which corresponds to a rate ratio of 1.22. 

In our third and final trial scenario we encode a moderate heterogeneous effect in order to challenge the prognostic adjustment method. For this scenario, in the trial DGP we encode the covariate means $u_1 = w_1 = 0$, $\eta = 1$ (heterogeneity present). We set $\zeta=0.057$ (which in this scenario also corresponds to a rate ratio of 1.22) in order to keep things as similar as possible to the additive scenario, modulo the heterogeneity. Under heterogeneity, $\zeta$ does not correspond to the rate ratio in the same way as before due to the extra term in the treatment-arm link function.

\begin{table}[htbp]
\centering
\caption{Trial scenarios used in the simulation study.}
\label{tbl:trial-scenarios}
\begin{tabular}{lccccc}
\hline
Scenario & $u_1$ & $w_1$ & $\eta$ & $\zeta$ & Rate ratio \\
\hline
Null effect           & 0 & 0 & 0 & 0    & 1.00 \\
Additive effect     & 0 & 0 & 0 & 0.2  & 1.22 \\
Heterogeneous effect                & 0 & 0 & 1 & 0.057 & 1.22 \\
\hline
\end{tabular}
\end{table}

\subsubsection{Modifications in Historical DGPs}

For each trial scenario, we also consider five different historical data-generating processes in order to understand how different estimators behave when the prognostic model is compromised. In practice, a prognostic score could be compromised by many different mechanisms: the distribution of covariates could be different between trial and historical populations, the relationship between the outcome and covariates could be different in different populations, the definitions of covariates and outcome could change from one dataset to another, a key covariate could be missing in the historical data, missing data in the historical sample could be imputed badly, the prognostic model could be misspecified, etc. Individually simulating each of these scenarios explicitly would be overwhelming and not very useful because they all result in the same thing: degraded predictive performance of the prognostic score in the trial population. Moreover, from a mathematical perspective, all of these real-world scenarios abstract to different distributions for the trial and historical data, regardless of the nature of the changes. Therefore in our simulations we focus only on explicit changes in covariate distributions (which can also change structural relationships when covariates are unmeasured). To capture the worst-case possible effects of any of these failure modes we include an estimator where the prognostic score is replaced with pure noise -- more on that below. 

Our historical data-generating process are based on the trial data-generating processes. For each trial scenario, we let $\eta$ and $\zeta$ be the same as in the trial data-generating process, i.e. these do not depend on $d$. However, we modify the historical covariate means $u_0$ and $w_0$ to be closer or further to their trial counterparts $u_1$ and $w_1$. As previously discussed, larger discrepancies between the historical and trial data-generating processes should diminish (but not necessarily eliminate) the efficiency gains of prognostic adjustment and also make sample size calculations unreliable. The five scenarios are given in \autoref{tbl:hist-scenarios}.

\begin{table}[htbp]
\centering
\caption{Historical data-generating process scenarios.}
\label{tbl:hist-scenarios}
\begin{tabular}{lcc}
\hline
Scenario & $u_0$ & $w_0$ \\
\hline
No shift & 0 & 0 \\
Small unobserved shift & 1.5 & 0 \\
Small observed shift   & 0   & 1.5 \\
Large unobserved shift & 3   & 0 \\
Large observed shift   & 0   & 3 \\
\hline
\end{tabular}
\end{table}

\subsection{Estimators}

We considered six different estimators of the rate ratio from the trial data, all of which can be described by different sets of adjustment covariates used in a GLM. That is, all of our estimators are found by the framework in \autoref{sec:est}. The only component that changes is the covariate set used to estimate the regression $\hat\mu$. The adjustment sets are described in \autoref{tbl:adjustment-sets}.


\begin{table}[htbp]
\centering
\caption{Adjustment sets considered in the simulation study.}
\label{tbl:adjustment-sets}
\begin{tabular}{p{0.22\linewidth} | p{0.44\linewidth} p{0.30\linewidth}}
\toprule
\bfseries Adjustment set & \bfseries Description & \bfseries Motivation \\
\specialrule{0.8pt}{3pt}{6pt}
None
& Unadjusted estimator.
& Simplest baseline. \\
\specialrule{0.4pt}{2pt}{2pt}
Covariates
& GLM adjusted for covariates $W$ (main terms).
& A strong, sensible baseline, commonly used. \\
\specialrule{0.4pt}{2pt}{2pt}
Noise + covariates
& GLM adjusted for $W$ and an additional noise variable.
& The worst-case setting for prognostic adjustment: when the prognostic model has no information. In practice, this could occur either because of bad specification of the prognostic model or mismatch between trial and historical populations, including not having access to key covariates in the historical population. \\
\specialrule{0.4pt}{2pt}{2pt}
Oracle + covariates
& GLM adjusted for $W$ and $\mu(W, A=0, D=1)$, the true control-arm conditional mean in the trial. We call this the ``oracle'' prognostic score.
& The best-case setting for prognostic adjustment: the prognostic function is known a priori without need of external data. \\
\specialrule{0.4pt}{2pt}{2pt}
Prognostic score
& GLM adjusted only for $\hat\rho(W)$, the prognostic score estimated from the external data.
& A variant of prognostic adjustment that is sensible if one is confident in the prognostic score. \\
\specialrule{0.4pt}{2pt}{2pt}
Prognostic score + covariates
& GLM adjusted for covariates $W$ and $\hat\rho(W)$, the prognostic score estimated from the external data.
& Our proposal for prognostic adjustment; the prognostic score is estimated from external data and the adjustment set includes covariates to hedge against a poor prognostic score. \\
\specialrule{0.9pt}{0pt}{0pt}
\end{tabular}
\end{table}

For adjustment in all estimators we used a GLM with log link and assuming a negative binomial outcome with dispersion parameter fixed to three. We did this so that the GLM would be mildly misspecified (the true outcome distribution is Poisson) but in pilot testing we saw essentially no difference in results for different values of the dispersion parameter, which makes sense since it minimally affects the fitting of the conditional mean. In all cases we computed standard errors according to \autoref{eq:asympvar} with 10-fold cross-fitting, although our results were essentially identical without the cross-fitting (see Appendix~\ref{app:sim_dpg_cv_no}).

In all cases we used a multivariate adaptive regression spline (MARS) learner \cite{Friedman1991-qh} to fit the prognostic score from the external data. We used the \texttt{tidymodels} wrapper for the \texttt{earth} package, setting \texttt{prod\_degree} to 3 and \texttt{num\_terms} to 50. We chose a single, fast-fitting learner rather than an ensemble so that the simulations could be easily replicated and run quickly on a single laptop. To counteract the possible loss in performance from using a single learner rather than an ensemble we deliberately chose an expressive learner that would be well-enough suited to learning the kinds of conditional means we encoded in the DGP: since we use mostly hinge functions, we expect MARS to be perform quite well. That said, the MARS learner is actually still slightly misspecified given the square term that appears in the true outcome regression.

To generate the ``noise'' for the noise + covariates estimator, we took the vector of in-trial predicted prognostic score values $\hat\rho(W)$ and shuffled the indices. This preserves the marginal distribution of the prognostic score while destroying any predictive power.

\noindent{}\textbf{Power Calculation}\\

Before estimating the rate ratio from trial data in each of our experimental runs, we first fit the prognostic score and use this to perform several power calculations according to the method described in \autoref{sec:power}. We calculated the prospective power for each of the estimators in \autoref{tbl:adjustment-sets}. The only difference in the power calculations was the value of $\hat\kappa_0$ used. For the unadjusted estimator, we used $\hat\kappa_0 = \hat\sigma_0$, the marginal standard deviation of the outcome in the external data. For the covariate adjusted estimator we used the cross-validated RMSE of a GLM (same settings as described above) fit to predict the outcome from the (main-terms) covariates based on the external data. For the prognostically adjusted estimators we used the cross-validated RMSE of the fit prognostic score $\hat\rho$ on the external data and for the oracle prognostically adjusted estimator we used the empirical RMSE of the oracle prognostic score $\mu(W,0,1)$ on the external data.

For all power calculations we used the true effect size as the target effect size. We used a desired type \RNum{1} error rate of 0.05 and a target power of 0.8.

\subsection{Experiments}

Given a combination of trial scenario and historical modification we first drew a historical dataset of size 2500 and estimated the prognostic score and the required sample size to hit 80\% power. We then drew trial datasets of sizes ranging from 100 to 400 and applied each of our estimators to each of the trial datasets. These sample sizes are to reflect a scenario where there is on the order of ten times more external data than trial data. We recorded the point estimates and estimated standard errors for each estimator. We repeated this process 1000 times for each combination of trial scenario and historical modification. 

\subsection{Results}

\noindent{}\textbf{Efficiency}\\

Our first question is whether prognostic adjustment increases efficiency relative to a baseline analysis where only (main terms) covariates are adjusted for, and how this changes across scenarios. To assess this, we computed the relative estimated efficiency of each estimator relative to the covariate adjusted GLM by dividing the estimated standard error for each estimator in each replicate of each scenario by the estimated standard error for the standard covariate adjusted estimator. This metric reflects the amount by which the confidence interval for each estimator would be smaller (or larger) than that of the covariate adjusted estimator. 
The relative efficiencies are shown in \autoref{fig:se} for $n = 250$ (trial sample size exactly one-tenth that of the external data).

\begin{figure}[!ht]
    \centering
    \includegraphics[width=1\textwidth, clip]{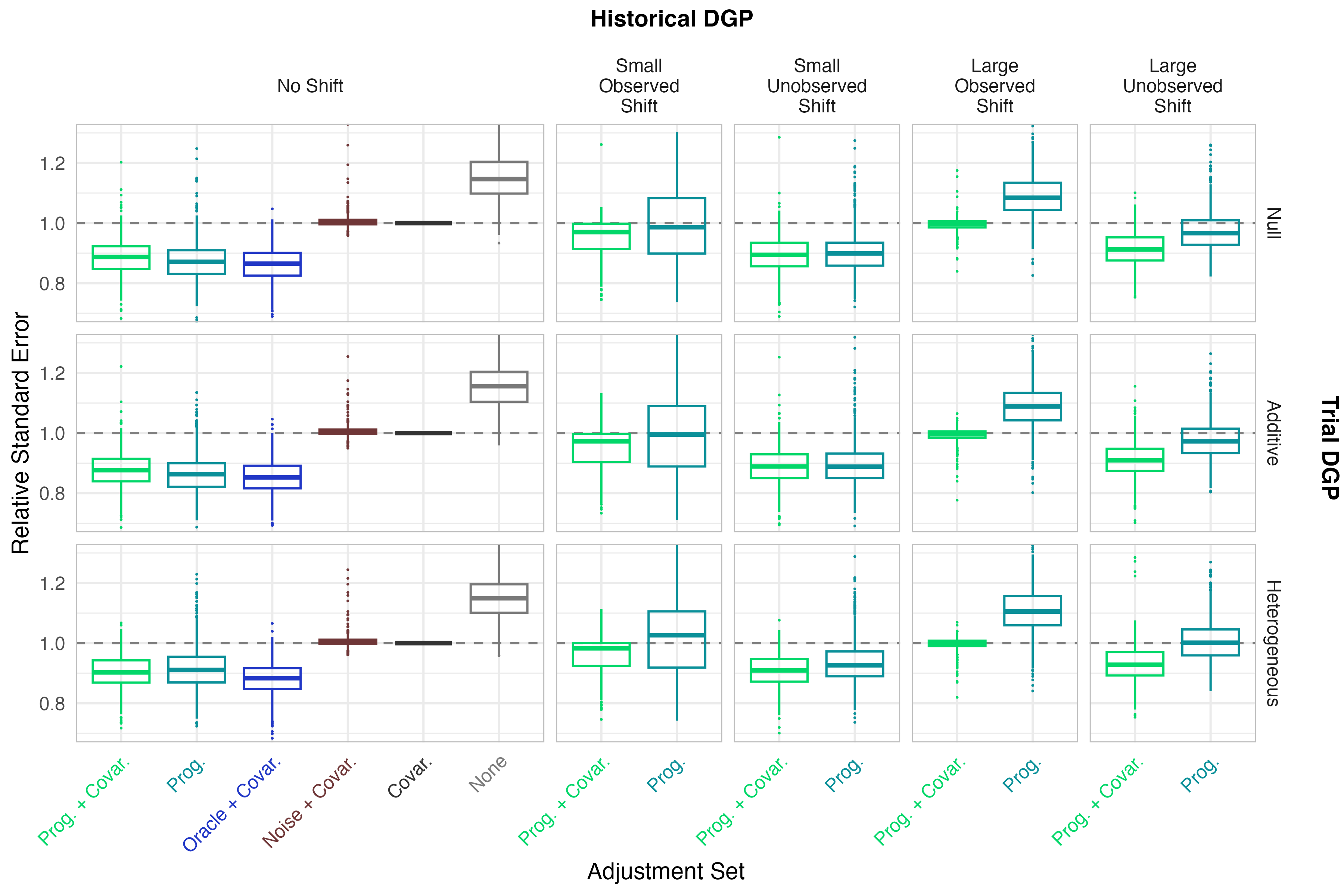}
    \caption{Estimated efficiencies of each estimator relative to standard covariate adjustment for $n = 250$. The box-and-whisker plot shows the distribution of the empirical efficiencies across replicates. Only the prognostically-adjusted estimators are shown for the shift scenarios since the other estimators do not make any use of the historical data and therefore do not change.}
    \label{fig:se}
\end{figure}

The figure shows that prognostic adjustment (with covariates) very, very rarely decreases efficiency by any meaningful amount, while in most cases it substantially improves it. The oracle prognostic score, which one would not have access to in practice, always increases efficiency and by a little bit more than the estimated prognostic score in our simulations.

When the prognostic model is pure noise the efficiency is generally unchanged relative to covariate adjustment, though there are a very small number of cases when the efficiency is more than 10\% worse. Thus the single added degree of freedom from adjusting for the prognostic score does not usually incur any substantial decrease in performance, even if the prognostic score is completely useless. This is meant to capture what could happen under a variety of potential catastrophic failures including misspecification of the prognostic score, changing definitions or omissions of covariates between the historical and trial datasets, drastic shifts between covariate distributions or their relationships to the outcome, etc. There is no need to simulate each of those scenarios explicitly because the impact is always the same: a prognostic score with decreased predictive ability in the trial, and no predictive ability at all in the worst-case scenario.

As expected, prognostic adjustment suffers in efficiency when there are larger differences between the historical and trial populations, either \textit{observed} covariate shifts or \textit{unobserved} covariate shifts (these are also structural changes in the outcome generating process, depending on perspective). Larger shifts mean worse performance, but prognostic adjustment with covariates almost never performs worse than covariate adjustment by itself, even in the most antagonistic scenarios. On the other hand, adjustment for \textit{only} the prognostic score can be much worse than covariate adjustment when the prognostic score is poorly estimated, and, in the limit, this estimator can perform as badly as the unadjusted estimator. The efficiency boost from prognostic adjustment is also slightly muted by effect heterogeneity.

\noindent{}\textbf{Coverage}\\
Naturally, it is very easy to for an estimator to have good empirical efficiency if it underestimates its own standard error. To assess this, we computed the empirical coverage of each of our estimators across scenarios, which is plotted in \autoref{fig:coverage}.

\begin{figure}[!ht]
    \centering
    \includegraphics[width=1\textwidth, clip]{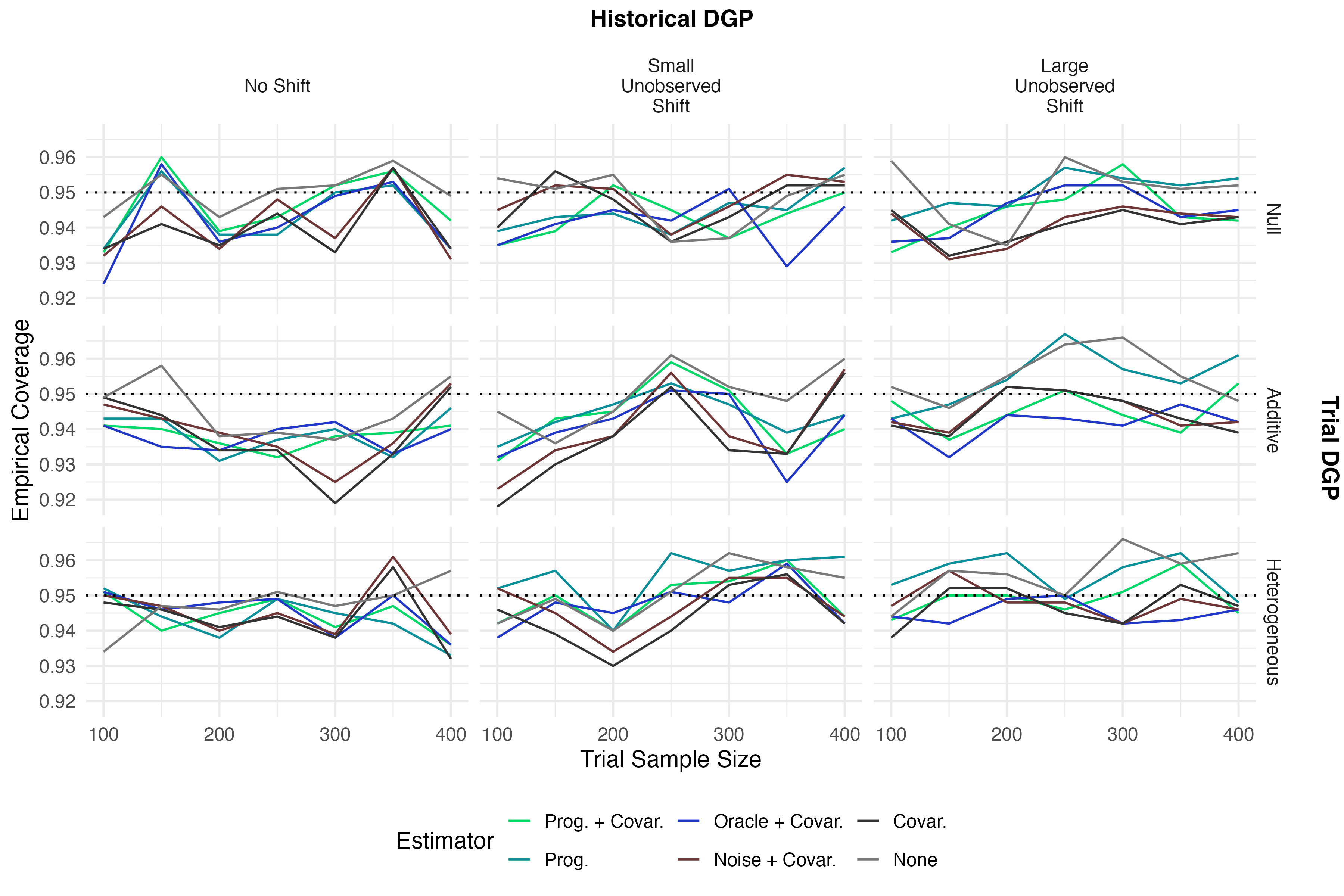}
    \caption{Empirical coverage of each estimator across scenarios as trial sample size varies.}
    \label{fig:coverage}
\end{figure}

Coverage is very close to nominal across sample sizes and scenarios for all estimators. Notably, coverage is still good even when the prognostic score is pure noise, demonstrating that prognostic adjustment is still safe as an analysis method regardless of population shifts, measurement errors, mismatched covariate and outcome definitions, etc. The good coverage shows both that none of our methods have any meaningful bias and that the standard errors are accurately estimated.

In Appendix~\ref{app:sim_dpg_cv_no} we show the equivalent coverage results \textit{without} cross-fit estimation of the standard errors. In that figure we see a very mild decrease in coverage for all the estimators that include the full set of covariates, which is why we recommend using cross-fitting for variance estimation. This is not specific to prognostic adjustment, since even the coverage for the standard covariate adjusted estimator improves with cross-fit standard errors.

\noindent{}\textbf{Sample Size Calculation}\\

Our second question is whether the sample size calculation presented above is empirically valid when the trial and external populations are identical and to what extent it fails when they are not. To assess this we computed the empirical proportion of significant results for each estimator across repetitions of each scenario and compared them to the estimated sample size we calculated for each estimator based on the historical data, averaged across repetitions. The results are shown in \autoref{fig:power}.

\begin{figure}[!ht]
    \centering
    \includegraphics[width=1\textwidth, trim={0mm 10mm 0mm 0mm}, clip]{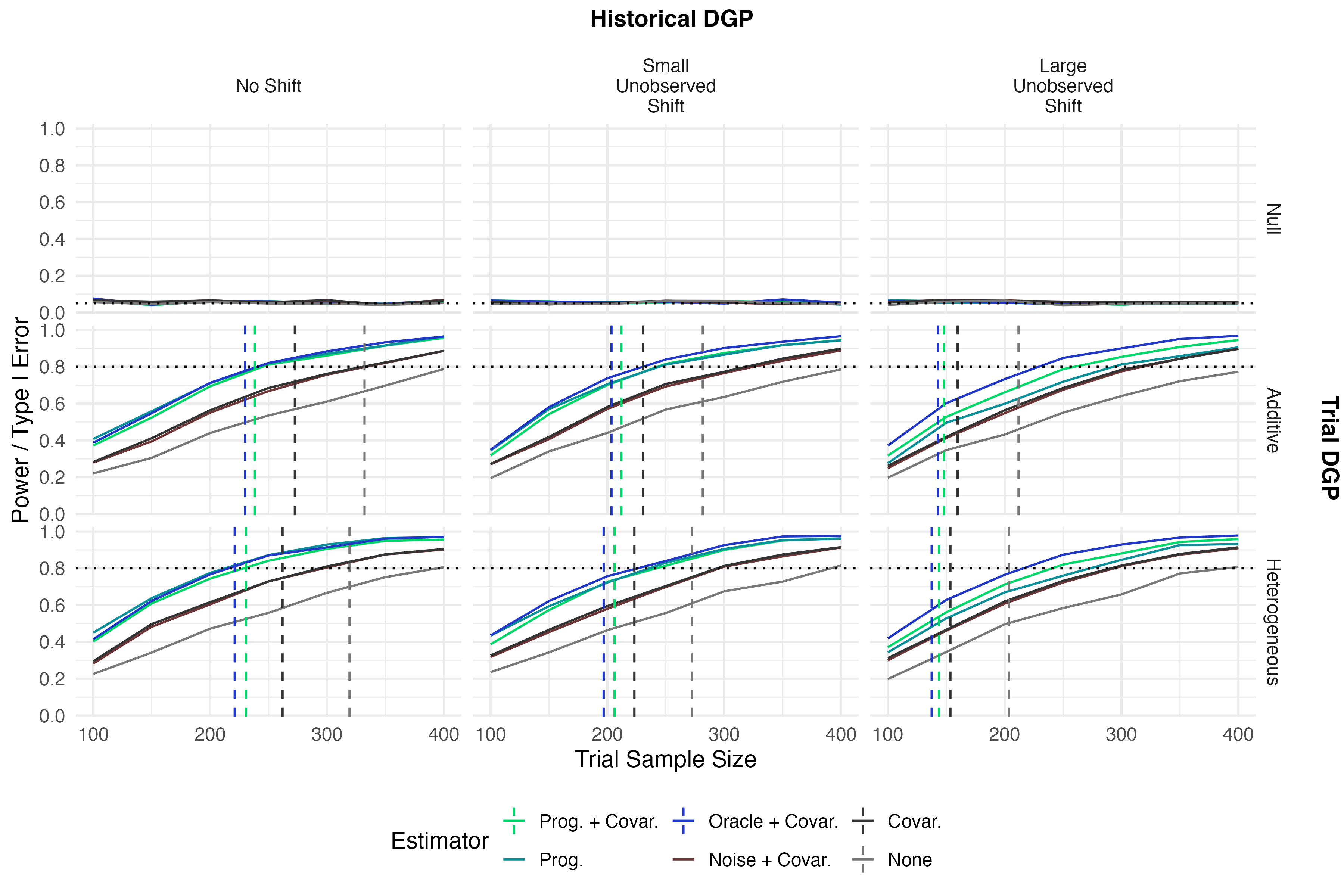}
    \caption{Empirical percentage of significant results (power, or type \RNum{1} error for the null scenario) of each estimator with increasing trial sample size across scenarios with no shift and unobserved covariate shifts. Vertical dashed lines indicate the average sample size estimated to attain 80\% power by each method. Equivalent plots for the observed covariate shifts show the same pattern so we omit them to reduce visual clutter. The estimated sample sizes when adjusting for only the prognostic score or for the noise prognostic score plus covariates are the same as that for prognostic adjustment with covariates.}
    \label{fig:power}
\end{figure}

The results confirm that the proposed sample size calculation is accurate when the required assumptions are met, i.e., when the trial and external populations are matched. Interestingly, our sample size calculation is slightly anti-conservative when the prognostic score and/or covariates are \textit{not} adjusted for, possibly due to the failure of asymptotic normality at these small sample sizes. However, this explanation is complicated by the fact that the unadjusted estimator has good coverage across scenarios. It is not entirely clear to us why accounting for the estimated or oracle prognostic score in the power calculation appears to resolve the issue. A more exact understanding of the sampling distribution could improve the performance of the unadjusted or covariate-adjusted sample size calculation but we leave this to future work because it does not directly pertain to prognostic adjustment.

The results also illustrate how the sample size calculation can fail when the trial and external populations are mismatched, even in ways that may not be immediately evident by examining observed covariate distributions (or comparing historical distributions to planned inclusion/exclusion criteria). Although smaller differences in the populations can be somewhat tolerated, large differences can lead to substantial underestimates of the required sample size. Overestimates are also theoretically possible, but we constructed the simulation to demonstrate the underestimate by making the shift result in more explainable variance than there actually is in the trial population.

Lastly, then null scenario results in \autoref{fig:power} show that all estimators achieve nominal type \RNum{1} error control. This is the case even when the prognostic score is pure noise (recall that this captures catastrophic failures due to population shifts, mismatched covariate/outcome definitions or omissions, etc.).

}

\section{Case study}

In this section, we examine the use of prognostic score adjustment for GLMs in RCTs involving people diagnosed with type 2 diabetes (T2D). Specifically, we explore the occurrence of hypoglycaemic episodes as defined by the American Diabetes Association (ADA) and investigate their implications for individuals diagnosed with type 2 diabetes (T2D). A hypoglycaemic episode can be classified in different categories, but in general, according to the ADA, it is characterized by a blood glucose level of less than 70 mg/dL, accompanied by symptoms that can include sweating, nausea, and confusion. Individuals with diabetes may experience hypoglycaemic episodes due to various factors, such as excessive insulin administration, insufficient food intake, or increased physical activity, leading to an imbalance in their glucose levels. The analyses presented here utilize data from Novo Nordisk A/S, drawn from 10 previously conducted RCTs in the field of diabetes, as detailed in \autoref{app_overview_cs}. For details on data missingness and imputation see \autoref{app_datmis}. Data from the same trials was used by L. Liao, E. Højbjerre-Frandsen, A. Hubbard and A. Schuler \cite{Liao2023} for a analysis of long term blood glucose level measured by  HbA1c using prognostic score adjustment for efficient estimators. Similarly, \citet{HojbjerreFrandsen2024} have done a prospective power estimation for a HbA1c analysis using linear adjustment for a prognostic score using the same trial data.

We conducted the analysis using data from the phase IIIb clinical trial designated NN9068-4229, which involved a participant population of insulin-naïve individuals diagnosed with T2D \citep{clinicalgov}. The participants in this trial were individuals whose diabetes was inadequately managed while receiving treatment with SGLT2 inhibitors, a class of oral anti-diabetic medications (OAD). Inadequately controlled was defined as having a HbA1C level ranging from 7.0\% to 11.0\% (inclusive). The trial was a 26-week, 1:1 randomized \textcolor{black}{comparing insulin IDegLira to insulin IGlar}, active-controlled, open-label, treat-to-target study that enrolled 420 participants. One participant was excluded due to lack of exposure to the trial product, resulting in a final sample size of $n=419$. Our related historical sample was derived from previously conducted RCTs that included a study population of insulin-naïve individuals with T2D who were inadequately controlled on their existing OADs. This historical sample comprised a total of approximately $\tilde{n} = 2286$ participants, all of whom received insulin IGlar.

Let RR be the rate ratio (IDegLira/IGlar) of the rate of any ADA classified hypoglycaemic episodes during the 26 weeks trial period. The number ADA classified hypoglycaemic episodes is the sum of hypoglycaemic episodes classified as either "Severe hypoglycaemia", "Documented symptomatic
hypoglycaemia", "Probable symptomatic
hypoglycaemia", "Asymptomatic hypoglycaemia" or "Pseudo hypoglycaemia" using the ADA definition. We test the $\He_0$-hypothesis of non-superiority of IDegLira against the alternative hypothesis of superiority as given by $\He_0: RR \geq 1$ and $\He_1: RR < 1$

\noindent{}\textbf{Prognostic score estimation}\\

To estimate the prognostic scores, a Discrete Super Learner was developed in accordance with recommended practices for effective machine learning outlined by \cite{FDA_ml}. This model possesses the oracle property, meaning it performs at the same level as the best machine learning algorithm available in the model library \cite{superlearner}. The Lasso machine learning model was chosen by the Discrete Super Learner for this analysis. The model achieved a root mean square error (RMSE) of 7.41 on the test data and 13.1 on the historical data. Further details regarding the prognostic score estimation can be found in \autoref{app:case_tuning}.

\noindent{}\textbf{Rate ratio estimation}\\

In this analysis we report the result of a similar analysis to that of the actual trial NN9068-4229 i.e. the GLM plug-in analysis using a Negative Binomial regression which adjust for treatment, region and log of the on treatment duration in years. The overdispersion parameter $r=\frac{\mu^2}{\sigma^2-\mu} = 0.54$ was estimated using the sample mean and standard deviation from NN9068-4229. We compare this to the proposed GLM plug-in analysis further adjusting for the prognostic scores. Additionally, we calculated the correlation between the fitted prognostic score and the trial outcome. The correlation for the prognostic score was found to be 0.514 for control subjects and 0.354 for treated subjects. These results suggest that adjusting for the prognostic score is likely to enhance estimation efficiency \cite{Schuler2020}. The results are summarised in \autoref{tab_res_full}.
\begin{table}[h]
\centering
\begin{tabular}{lccc}
{Prognostic score} & GLM plug-in analysis \\
\hline
with & 0.906 (s.e. 0.113\textcolor{black}{, CI [0.683, 1.128]}) \\
without & 0.907 (s.e. 0.123\textcolor{black}{, CI [0.667, 1.148]}) \\
\end{tabular}
\par\medskip
\caption{Estimates for rate ratio (IDegLira/IGlar) of the rate of any ADA classified hypoglycaemic episodes during the 26 weeks trial period with IF estimated standard deviation found using CV. \textcolor{black}{Without using CV the s.e. is 0.112 and 0.121 with and without prognostic score adjustment, respectively.} The total sample size is $n=419$. }
\label{tab_res_full}
\end{table}

From \autoref{tab_res_full}, we see that the smallest estimated standard deviations is obtained using the GLM plug-in with prognostic score. The two adjustment methods obtain similar point estimates. This amount in a change in post hoc power estimation of approximately 2.5\% and a relative power increase of 11.5\%.

\section{Discussion}
{\color{black}
Our theory and simulations show that prognostic adjustment with GLMs can improve the efficiency of trial analyses, even when trial and external populations are not identical. Moreover we develop a sample size calculation, demonstrate its accuracy when trial and external populations match, and highlight its failures when they do not.
}

Compared to other historical borrowing methods, prognostic covariate adjustment theoretically guarantees type \RNum{1} error rate control and confidence interval coverage in general settings. Specifically, using the GLM plug-in procedure  described by \citet{RosenblumvanderLaan+2010} with the variance estimate found from the IF we have consistency of the marginal effect estimate and asymptotic control of the type \RNum{1} error rate even when the model is misspecified under an RCT. In this sense, using GLM prognostic adjustment for trial analysis is ``assumption-free''. 

Efficiency gains from prognostic adjustment may be limited if highly prognostic baseline covariates are already directly adjusted for, or if there is not any substantial nonlinearity left to model after standard adjustment. However, our theory and simulations show that, even in the worst-case scenario, prognostic adjustment doesn't hurt as long as prognostic covariates are also adjusted for in the analysis. This is our recommendation, which aligns with the FDA's opinion \cite[p. 3]{FDA_cov}: "\textit{Covariate adjustment is acceptable even if baseline covariates are strongly associated with each other (e.g., body weight and body mass index). However, adjusting for less correlated baseline covariates generally provides greater efficiency gains.}".

Prognostic adjustment with GLMs offers several practical advantages over data fusion approaches. For instance, it does not require a single, well-defined treatment in the historical data. Moreover, it does not necessitate an exact overlap of the covariates measured in the historical and trial datasets. As long as there is some overlap and informativeness, the historical data can be leveraged effectively. {\color{black} This is evident in our simulation settings under covariate shifts, which also capture what would happen if covariates were defined differently in different populations. Even in the worst case when the prognostic score is not predictive at all due to shifts, omissions, or definitional errors, GLM prognostic score adjustment (including covariates) very rarely performs worse in variance than standard covariate adjustment.}

{\color{black}We also propose a corresponding method of sample size calculation for prognostically adjusted GLMs, which, like any other, depends strongly on assumptions and has failure points. Studies may be underpowered if the assumption that $\tau>0$ fails, or even if the marginal outcome variance is underestimated. The largest risk that is specific to prognostic adjustment is overestimating the predictive ability of the prognostic model. For this reason, it is important to always be aware of the tradeoffs, perform sensitivity analyses by setting parameters conservatively enough to tolerate inaccuracies, and perform thorough simulation testing. 

There is also no need to perform a single power analysis: the required sample sizes from a prognostically adjusted GLM, standard GLM, and unadjusted analysis may be compared. If there is no substantial sample size reduction from the former, there is no need to justify the additional assumptions. If there is a difference, and one is comfortable with, e.g., the assumptions for the unadjusted calculation, then the size of this gap might be used to justify further simulation studies or literature reviews to verify or falsify the additional assumptions for the adjusted analyses.} 

There has been significant growth in the availability and performance of nonlinear regression modeling, particularly in the area of deep learning. The intersection of this technological development with the creation of large historical control databases provides an opportunity to use prognostic covariate adjustment to substantially improve future clinical trials. However, in some clinical indications like rare disease studies the amount of historical data poses a challenge for using adjustment with a prognostic score, since there should be sufficient historical data to split into training and testing datasets. Another possibility in those cases, would be to utilized prognostic adjustment with pre-built or public prognostic models. The analyst does not need direct access to the historical data if they can query a model for predictions. This is particularly useful in cases where data cannot be moved, such as when privacy must be protected or data has commercial value. Individual subject data from the historical sample is not necessary to addend the trial data with the prognostic score.

When applying linear adjustment with a prognostic score in GLMs, it is crucial to consider that adjusting the sample size for the primary endpoint may lead to a reduction in power for secondary analyses. This is particularly significant in disease areas where a minimum number of participants must be exposed. However, this issue can be mitigated by incorporating a prognostic score for secondary endpoints to enhance power for these analyses. 
However, developing distinct models for each analysis is logistically challenging and resource-intensive, especially when dealing with numerous secondary endpoints. Therefore, we recommend using the method selectively for secondary endpoints of high clinical importance. Sponsors should consider the risk of underpowered secondary and subgroup analyses. Alternatively, sponsors could maintain the same sample size but increase power by using the method, which may be more readily accepted.

\subsection*{Acknowledgement}
We thank Martin Theil Hansen for curating the case study dataset. 
This research was supported by Innovation Fund Denmark (grant number 2052-00044B). 

\bibliographystyle{plainnat}
\bibliography{litterature}

\newpage

\appendix

\section{. \hspace{0.1cm} Negative Binomial working model}\label{app:NB}
In this appendix we will go through the argument by \citet{RosenblumvanderLaan+2010} to assess if Theorem 1 of \cite{RosenblumvanderLaan+2010} is also valid when the working model is a Negative Binomial regression. Firstly, the density of the Negative Binomial distribution for $k=0, 1, 2, \cdots$ is
\begin{align}\label{eq:exp_fam}
\begin{split}
    f(k; r, p) &= \binom{k+r-1}{k}(1-p)^k p^r\\
    &= \text{exp}\left(\text{ln}\left(\binom{k+r-1}{k}(1-p)^k p^r\right)\right)\\
    &= \binom{k+r-1}{k}\text{exp}\left(\text{ln}\left((1-p)^k p^r\right)\right)\\
    &= \binom{k+r-1}{k}\text{exp}\left(k\text{ln}\left(1-p\right) +r\text{ln}\left(p\right)\right)\\
    &= \text{exp}\left(Y\eta - b(\eta) + c(Y, \phi)\right),
\end{split}
\end{align}
where $Y=k$ is the outcome, $r$ is the number of successes until the experiment is stopped and $p\in [0,1]$ is the probability of success in each experiment. The last equality shows that with $r$ known the Negative Binomial distribution is a part of the exponential family written in canonical form with $\sum_j \beta_j f_j(W, A) = \eta = \text{ln}(1-p)$ and $b(\eta) = -r \text{ln}(p)$. \citet{Bickel2001} prove that for the canonical link the mean can be expressed as $\mu= g^{-1}(\eta) = \frac{\partial}{\partial \eta} b(\eta)$. For the Negative Binomial distribution this means that 
\begin{align}
\begin{split}
    \mu = \frac{\partial}{\partial \eta} -r\text{ln}(p) = -\frac{\partial}{\partial \eta} r\text{ln}(1-\text{exp}(\eta))= r \frac{\text{exp}(\eta)}{1-\text{exp}(\eta)} = r \frac{1-p}{p},
\end{split}
\end{align}
where we use $\eta = \text{ln}(1-p) \Leftrightarrow 1-\text{exp}(\eta)=p$. This means that the canonical link function is 
\begin{align}
\begin{split}
    g(\mu) = \text{ln}\left(\frac{\mu}{r+\mu} \right).
\end{split}
\end{align}
The expected log-likelihood is strictly concave since $-\frac{\partial^2}{\partial \eta^2} r\text{ln}(1-\text{exp}(\eta))=r \frac{\text{exp}(\eta)}{(1-\text{exp}(\eta))^2}>0$ for all values of $\eta$. This means that that there either exists a unique maximizer $\beta^*$ of the expected log-likelihood or that the euclidean norm of the MLE grows without bound a sample size goes to infinity, which would be detected based on the assumption that each component of $\hat\beta$ does not exceed the threshold $M$ \cite{RosenblumvanderLaan+2010}.

\paragraph{}
We now let the working model distribution be $p_0(W, A, Y)=p_{0}(Y| W, A)\pi_A p_0(W)$, where $p_{0}(Y| W, A)$ is now the density using the pre-specified Negative Binomial model with the parameters estimated by MLE and $p_0(W)$ is the empirical distribution of $W$. The derivative of the log-likelihood is $0$ at $\hat{\beta}$. This means that for the i'th term of the log-likehood we have
\begin{align}\label{eq:score_nb}
\begin{split}
   \frac{\partial}{\partial \beta} \ell (\eta(\beta)) = \left(Y - b'(\eta)\right)\cdot   \frac{\partial \eta}{\partial \beta} = \left(Y - \E_{p_0}[Y| W, A]\right)\cdot   \frac{\partial \eta}{\partial \beta} = 0,
\end{split}
\end{align}
where we have abused the notation slightly by using $p_0$ to denote the assumed density for $Y| W, A$. 
Since we included an intercept and the exposure variable in the linear regressor we now have 
\begin{align}
\begin{split}
   \sum_{i=1}^n\left(Y_i - \E_{p_0}[Y| W_i, A_i]\right) &= 0\\ 
   \sum_{i=1}^nA_i\left(Y_i - \E_{p_0}[Y| W_i, A_i]\right) &= 0. 
\end{split}
\end{align}
Since these are both 0 and the propensity scores $\pi_a$ are known we can multiply with $1/\pi_a$ (using the positivity assumption) to obtain the first term of the IF in \eqref{eq:IF_mean}. The last term of the IF is also 0, which can be seen by using the law of total expectation and that under $p_{0}(W)$ each observation has $\Pe(W | A) = \Pe(W) =\frac{1}{n}$ due to randomization, we have
\begin{align}
\begin{split}
  &\frac{1}{n} \sum_{i=1}^n\E[Y|W_i, A=a] = \frac{1}{n} \cdot n \E[Y|A=a] \\
  \bimp& \frac{1}{n} \sum_{i=1}^n\left(\E[Y|W_i, A=a] \E[Y|A=a]\right) = 0. 
\end{split}
\end{align}
Thus when using the Negative Binomial regression as our working model the scores solve the IF equation i.e. setting \eqref{eq:IF_mean} equal to 0. We can then use the exact same argument as \citet{RosenblumvanderLaan+2010} to conclude that the procedure is a TMLE procedure done in zero steps using the Negative Binomial working model. 

Now we need to show that the five conditions in Theorem 1 of \citet{vanderLaanRubin+2006} still hold when using the Negative Binomial regression model. We first note that the first two conditions relate to the statistical model $\mathcal{M}$, for which we have the same assumption of a semi-parametric model due to randomization. For the third condition we can use the exact same calculations as \citet{RosenblumvanderLaan+2010} replacing $1/2$ with $\pi_a$ for generality. These calculations rely heavily on the plug-in bias being 0 since the IF equation is solved and the independence of $W$ and $A$. The fifth condition can be shown using the exact same argument as \citet{RosenblumvanderLaan+2010}, which rely on the parameters $\hat{\beta}$ converging towards $\beta^*$ and the consistency of the counterfactual means $\hat{\Psi}_a \to_P \Psi_a$. 

So we now need to check the fourth condition to complete the proof. Firstly we need to check which conditions should be assumed for $\hat\phi_a$ to be in a Donsker class for $a\in\{0,1\}$. Since the IF is in a parametric class indexed by $\beta$ (again using that the propensity scores are known) we can use example 19.7 from \citet{van2000asymptotic}. This example shows that for a parametric class to be a Donsker class it suffices to show a type of Lipschitz continuety in $\L_2$, i.e. there exists a measurable function $m$ such that
\begin{align}
\begin{split}
  ||\hat\phi_{a, \beta}(W, A, Y)-\hat\phi_{a, \Tilde{\beta}}(W, A, Y)||_{\L_2}\leq m(W, A, Y)||\beta-\Tilde{\beta}||_{\L_2},
\end{split}
\end{align}
for all allowed values of $\beta$. This condition would for example be fulfilled if one assumed a bounded parameter space for the $\beta$ values which also should contain the true $\beta^*$. Instead of assuming this Lipschitz condition one can also assume that the the sectional variation norm of $\hat{\Psi}_{a, \beta}$ is uniformly bounded, and that could be achieved showing all the first order mixed derivatives w.r.t. continuous $W$ components are bounded in absolute value.

\section{. \hspace{0.1cm} Reduced form of the asymptotic variance }\label{app:reducedformasympvar}
In this appendix we will derive the expression in \eqref{eq:var_pop}. First we observe that, 
\begin{align}
\begin{split}
    v_\infty^2 &= \var\left[r_0'\phi_0 + r_1'\phi_1\right]\\
    &= r_0'^{\, 2}\var[\phi_0]+r_1'^{\, 2}\var[\phi_1]+2r_0'r_1'\cov[\phi_0, \phi_1]\\
    &= r_0'^{\, 2}\var[\phi_0]+r_1'^{\, 2}\var[\phi_1]-2|r_0'r_1'|\cov[\phi_0, \phi_1], 
\end{split}
\end{align}
where the sign change is guaranteed by the assumption in \eqref{eq:asump_r}. Then we just need to evaluate the variances and covariances and assemble these. For this we will make use of the fact that $1_a(A)Y = 1_a(A)= 1_a(A)Y(A) = 1_a(A)Y(a)$ to replace outcomes with potential outcomes. We will also that we are in the setting of an RCT meaning that $A\indep Y(a)$ and $A\indep W$. 

\subsection{Evaluating \texorpdfstring{$\var[\phi_a]$}{Var[phi-a]}}
\begin{align}
\begin{split}
   \var[\phi_a] &= \var\left[\frac{1_a(A)}{\pi_a}(Y-\mu^*(W, a))+(\mu^*(W, a)-\Psi_a) \right]\\
   &= \var\left[\frac{1_a(A)}{\pi_a}(Y-\mu^*(W, a)) \right] + \var\left[\mu^*(W, a) \right] + 2\cov\left[\frac{1_a(A)}{\pi_a}(Y-\mu^*(W, a)), \mu^*(W, a)\right]\\
   &= \E\left[\left(\frac{1_a(A)}{\pi_a}(Y(A)-\mu^*(W, a)) \right)^2\right] - \E\left[\frac{1_a(A)}{\pi_a}(Y(A)-\mu^*(W, a)) \right]^2 + \var\left[\mu^*(W, a) \right] \\&\quad+ 2\cov\left[\frac{1_a(A)}{\pi_a}Y(A), \mu^*(W, a)\right]- 2\cov\left[\frac{1_a(A)}{\pi_a}\mu^*(W, a), \mu^*(W, a)\right]\\
   &= \E\left[\frac{1_a(A)}{\pi^2_a}\left(Y(A)-\mu^*(W, a) \right)^2\right] - \left(\Psi_a -\Psi_a\right)^2 + \var\left[\mu^*(W, a) \right] \\&\quad+ 2\cov\left[Y(A), \mu^*(W, a)\right]- 2\var\left[\mu^*(W, a)\right]\\
   &= \frac{1}{\pi_a}\E\left[\left(Y(A)-\mu^*(W, a) \right)^2\right] - \var\left[\mu^*(W, a) \right] + 2\cov\left[Y(A), \mu^*(W, a)\right]- \var\left[Y(a)\right]+\var\left[Y(a)\right]\\
   &= \frac{1}{\pi_a}\kappa_a^2 -  \var\left[Y(A)-\mu^*(W, a)\right]+\var\left[Y(a)\right]\\
   &= \frac{1}{\pi_a}\kappa_a^2 - \E\left[\left(Y(A)-\mu^*(W, a)\right)^2\right]+\sigma_a^2\\
   &= \frac{1}{\pi_a}\kappa_a^2 - \kappa_a^2 +\sigma_a^2\\
   &= \frac{1-\pi_a}{\pi_a}\kappa_a^2 +\sigma_a^2
\end{split}
\end{align}

\subsection{Evaluating \texorpdfstring{$\cov[\phi_0, \phi_1]$}{Cov[phi-0, phi-1]}}
Here we will use that $1_0(A)\cdot 1_1(A)=0$. 
\begin{align}
\begin{split}
    &\cov[\phi_0, \phi_1]\\
    =& \cov\left[\frac{1_0(A)}{\pi_0}(Y-\mu^*(W, 0))+(\mu^*(W, 0)-\Psi_0) , \frac{1_1(A)}{\pi_1}(Y-\mu^*(W, 1))+(\mu^*(W, 1)-\Psi_1) \right]\\
    =& \cov\left[\frac{1_0(A)}{\pi_0}Y(0), \frac{1_1(A)}{\pi_1}Y(1)\right]
    +\cov\left[\frac{1_0(A)}{\pi_0}Y(0), \left(1-\frac{1_1(A)}{\pi_1}\right) \mu^*(W, 1)\right]\\
    &+ \cov\left[\frac{1_1(A)}{\pi_1}Y(1), \left(1-\frac{1_0(A)}{\pi_0}\right) \mu^*(W, 0)\right] +
    \cov\left[\left(1-\frac{1_1(A)}{\pi_1}\right) \mu^*(W, 1), \left(1-\frac{1_0(A)}{\pi_0}\right) \mu^*(W, 0)\right]\\
    =& -\Psi_0\Psi_1+\E[Y(0)\mu^*(W, 1)]+\E[Y(1)\mu^*(W, 0)]-\E[\mu^*(W, 0)\mu^*(W, 1)]\\
    =& \E[Y(0)\mu^*(W, 1)]-\Psi_0\Psi_1
    +\E[Y(1)\mu^*(W, 0)]-\Psi_0\Psi_1
    -\left(\E[\mu^*(W, 0)\mu^*(W, 1)]-\Psi_0\Psi_1\right)\\
    =& \cov[Y(0),\mu^*(W, 1)]+ \cov[Y(1),\mu^*(W, 0)] - \cov[\mu^*(W, 0),\mu^*(W, 1)]\\
    =& \cov[Y(0),\mu^*(W, 1)]+ \cov[Y(1),\mu^*(W, 0)] - \cov[\mu^*(W, 0),\mu^*(W, 1)] - \cov\left[Y(0),Y(1)\right]+\cov\left[Y(0),Y(1)\right]\\
    =& \cov\left[Y(0),Y(1)\right] - \cov\left[Y(0)- \mu^*(W, 0),Y(1)- \mu^*(W, 1)\right]\\
    =& \tau \sigma_0\sigma_1-\eta\kappa_0\kappa_1.
\end{split}
\end{align}
The last equation comes from the definition of the correlation $\cor(x,y)=\frac{\cov(x,y)}{\sqrt{\var(x)\var(y)}}$ applied to each term. 

\subsection{Special case - unadjusted estimator}
Let our estimand of interest be the ATE $\Psi = \Psi_1 -\Psi_0$ and assume that we use a linear model (GLM with normal density and identity as link function) with no covariate adjustment i.e. an ANOVA model. In this case the estimator would just be the difference of the two sample means and thus $\mu^*(W, a) = \Psi_a$ and $\kappa_a = \sigma_a$. This implies that the covariance term $\tau \sigma_0\sigma_1-\eta\kappa_0\kappa_1=0$. Thus, 
\begin{align}
    v_\infty^2 = \left( \frac{\pi_1}{\pi_0}\sigma_0^2 + \sigma_0^2\right) + \left(\frac{\pi_1}{\pi_0}\sigma_1^2 + \sigma_1^2\right)=  \frac{\sigma_0^2}{\pi_0} + \frac{\sigma_1^2}{\pi_1},
\end{align}
which is precisely the know variance of the unadjusted estimator. 

\subsection{Special case - unidimensional linear regression}
Let our parameter of interest be the ATE $\Psi = \Psi_1 -\Psi_0$ and presume
we use a linear GLM adjusted for a single covariate and a treatment-covariate interaction term. Standard theory (see \cite{HojbjerreFrandsen2024} or \cite{Schuler2020}) states that $\mu^*(W, a)=\var(W)^{-1}\cov(W^\top,Y(a))w$. Thus some algebra shows $\kappa_a^2=\E\left[(Y(a)-\var(W)^{-1}\cov(W^\top,Y(a))W)^2\right]= (1-\rho_a^2)\sigma_1^2$. Further algebra gives
\begin{align}
    v_\infty^2 =\frac{\sigma_0^2}{\pi_0} + \frac{\sigma_1^2}{\pi_1} - \pi_1\pi_0\left(\frac{\rho_0\sigma_0}{\pi_0}+ \frac{\rho_1\sigma_1}{\pi_1}\right)^2,
\end{align}
which shows that the result of Schuler \cite{Schuler+2022+151+171} is a special case of \eqref{eq:var_pop}.

\section{. \hspace{0.1cm} Supporting Lemmas and Proofs for Theorem 1}\label{app:proof}

In this section we will state and proof results leading to the establishment of \autoref{thm:main}. We start by defining a generic observation from the current RCT trial using an additional covariate $\rho(W)$ as $O_{\rho}\coloneqq(W, \rho(W), A, Y)$, where $\rho: \mathcal{W} \to \R$. To denote the $i$'th observation we write $o^i_{\rho}$. Then $O_{\mu_0}\coloneqq(W, \mu(W, 0), A, Y)$ is an observation using the true conditional mean outcome function. Let the score function of the GLM be denoted by $s(o_\rho; \beta)=  \frac{\partial }{\partial \beta} \ell(o_\rho; \beta)$. Let $\mathscr{L}(\rho ;\beta)=\E_\mathcal{P}[\ell(O_\rho; \beta)]$ and $\mathscr{S}(\rho; \beta)=\E_\mathcal{P}[s(O_\rho; \beta)]$ be the expected log-likelihood and score, respectively, at $\beta$ under a given distribution $\mathcal{P}$ for $O_\rho$. We again denote the new RCT data by $\P_n$ and the historical data by $\P_{\widetilde{n}}\,$ sampled from the same distribution $\mathcal{P}$.

\subsection{Supporting lemmas}

\begin{lemma}\label{lemma:convergence_beta}
    Let $\hat\rho_{\tilde n}\, : \mathcal{W} \to \mathcal{Y}$ be a function learned from  $\P_{\widetilde{n}}$. Let $\beta^*_0$ maximize $\mathscr{L}(\rho_0, \beta)=\E_\mathcal{P}[\ell(O_{\rho_0}\, ; \beta)]$ over a compact parameter space $\{\beta : |\beta|\leq b\}$ for a fixed function $\rho_0\, : \mathcal{W} \to \mathcal{Y}$. Let $\beta_{\widetilde{n}, \, n}$ maximize $\mathscr{L}_n(\hat\rho_{\tilde n}, \beta)= \frac{1}{n}\sum_{i=1}^n \ell(o^i_{\hat\rho_{\tilde n}}; \beta)$. Assume that the log-likelihood $\ell$ fulfills a $m(\beta)$-Lipschitz condition in the data,
    \begin{align*}
        |\ell(o, \beta)-\ell(o', \beta)|\leq m(\beta)||O-O'||_{\mathcal{L}_2},
    \end{align*}
    for all realisations $o$ and $o'$ and where $m$ is bounded on the domain of $\beta$. Then $\left|\hat\rho_{\tilde n}\,(W) - \rho_0(W)\right|\overset{L^2}{\longrightarrow} 0$ as $\widetilde{n}\to\infty$ implies $\beta_{\widetilde{n}, \, n}\pto \beta^*_0$ for $n, \widetilde{n} \to \infty$. 
\end{lemma}

\begin{proof}
We consider Theorem 5.7 from \citet{Vaart_1998} using $\mathscr{L}_n(\hat\rho_{\tilde n}, \beta)= \frac{1}{n}\sum_{i=1}^n \ell(o^i_{\hat\rho_{\tilde n}}; \beta)$ as a sequence of random functions and $\mathscr{L}(\rho_0, \beta)=\E_\mathcal{P}[\ell(O_{\rho_0}\, ; \beta)]$ as a fixed function of $\beta$. The theorem states that $\beta_{\widetilde{n}, \, n}\pto \beta^*_0$ for $n, \widetilde{n} \to \infty$ if for every $\eps > 0$
\begin{align}
    &\underset{\{\beta : |\beta|\leq b\}}{\text{sup}} |\mathscr{L}(\rho_0, \beta) - \mathscr{L}_n(\hat\rho_{\tilde n}, \beta) | \pto 0, \\
    &\underset{\{\beta : ||\beta - \beta^*_0||\geq \eps\}}{\text{sup}} \mathscr{L}(\rho_0, \beta) < \mathscr{L}(\rho_0, \beta_0^*).
\end{align}
The second condition is always fulfilled due to the strict concavity of the expected log-likelihood and since $\beta^*_0$ is the maximizer of this. For the first condition we first observe using the triangle inequality that,
\begin{align*}
    \underset{\{\beta : |\beta|\leq b\}}{\text{sup}} |\mathscr{L}(\rho_0, \beta) - \mathscr{L}_n(\hat\rho_{\tilde n}, \beta) | \leq \underset{\{\beta : |\beta|\leq b\}}{\text{sup}} |\mathscr{L}(\rho_0, \beta) - \mathscr{L}_n(\rho_0, \beta) | + \underset{\{\beta : |\beta|\leq b\}}{\text{sup}} |\mathscr{L}_n(\rho_0, \beta) - \mathscr{L}_n(\hat\rho_{\tilde n}, \beta) |,
\end{align*}
where $\mathscr{L}_n(\rho_0, \beta)= \frac{1}{n}\sum_{i=1}^n \ell(o^i_{\rho_0}; \beta)$. The first term goes to 0 in probability due to the law of large numbers. For the second term we use the triangle inequality and the $m(\beta)$-Lipschitz condition, to get 
\begin{align*}
    \underset{\{\beta : |\beta|\leq b\}}{\text{sup}} |\mathscr{L}_n(\rho_0, \beta) - \mathscr{L}_n(\hat\rho_{\tilde n}, \beta) | &= \underset{\{\beta : |\beta|\leq b\}}{\text{sup}} \left|\frac{1}{n}\sum_{i=1}^n  \ell(o^i_{\rho_0}; \beta) - \ell(o^i_{\hat\rho_{\tilde n}}; \beta) \right|\\
    &\leq \underset{\{\beta : |\beta|\leq b\}}{\text{sup}} \frac{1}{n}\sum_{i=1}^n \left|  \ell(o^i_{\rho_0}; \beta) - \ell(o^i_{\hat\rho_{\tilde n}}; \beta) \right|\\
    &\leq \underset{\{\beta : |\beta|\leq b\}}{\text{sup}} \frac{1}{n}\sum_{i=1}^n m(\beta) ||O_{\rho_0}-O_{\hat\rho_{\tilde n}}||_{L_2}\\
    &= \underset{\{\beta : |\beta|\leq b\}}{\text{sup}} \frac{1}{n}\sum_{i=1}^n m(\beta) |\rho_0(W_i)-\hat\rho_{\tilde n}(W_i)|,
\end{align*}
where we in the last equality use the definition of the $L_2$ norm. Since $m$ is bounded on the domain of $\beta$ we know that the last term is smaller than the bound on $m$. Thus, we now only need to show $ \frac{1}{n}\sum_{i=1}^n |\rho_0(W_i)-\hat\rho_{\tilde n}(W_i)|\pto 0$. Since convergence in probaility is implied by $L_1$ convergence we can instead show $\E\left[\frac{1}{n}\sum_{i=1}^n |\rho_0(w_i)-\hat\rho_{\tilde n}(w_i)|\right] \to 0$. Using the law of total expectation we get
\begin{align*}
    \E\left[\frac{1}{n}\sum_{i=1}^n |\rho_0(W_i)-\hat\rho_{\tilde n}(W_i)|\right] &= \E\left[\E\left[\frac{1}{n}\sum_{i=1}^n |\rho_0(W_i)-\hat\rho_{\tilde n}(W_i)|\, \big|\, \P_{\widetilde{n}}\right]\right]\\
    &= \E\left[ \E\left[|\rho_0(W)-\hat\rho_{\tilde n}(W)|\, \big|\, \P_{\widetilde{n}}\right]\right]\\
    &= \E\left[ |\rho_0(W)-\hat\rho_{\tilde n}(W)|\right]\\
    &= || \rho_0(W)-\hat\rho_{\tilde n}(W)||_{L_1}\\
    &\leq || \rho_0(W)-\hat\rho_{\tilde n}(W)||_{L_2}
\end{align*}
where we in the second equality use that when the historical data is given the terms in the sum have the same mean value since $\hat\rho_{\tilde n}$ becomes a fixed function. Now using the assumption that $\left|\hat\rho_{\tilde n}\,(W) - \rho_0(W)\right|\overset{L^2}{\longrightarrow} 0$ as $\widetilde{n}\to\infty$ completes the proof. 
\end{proof}

\begin{lemma}\label{lemma:mean_0_star}
    The IF $\phi^*_\Psi$ in \eqref{eq:IF} for the marginal effect estimator $\hat{\Psi}$ has mean 0.
\end{lemma}

\begin{proof}
    In this proof we will show that the consistency of the estimated counterfactual means, $\hat{\Psi}_a$ for $a\in \{0, 1\}$, implies $\E[\phi^*_\Psi(O)]=0$. We start by using \eqref{eq:consistency} and observe that the mean of the IF for the counterfactual outcome mean is 0, since for both values of $a$ we have
\begin{align}
\begin{split}
    \E[\phi_a^*(O)]&=\E\left[\frac{1_a(A)}{\pi_a}(Y-\mu^*(W, a))+(\mu^*(W, a)-\Psi_a)\right] \\
    &= \E\left[\frac{1_a(A)}{\pi_a}(Y(a)-\mu^*(W, a))\right] \\
    &= \E\left[Y(a)-\mu^*(W, a)\right] \\
    &=0.
\end{split}
\end{align}
Thus, we can determine the mean of the IF for the marginal effect estimator
\begin{align}
\begin{split}
    \E[\phi_\Psi^*(O)]&=\E\left[r_0'(\Psi_1, \Psi_0)\phi^*_0 + r_1'(\Psi_1, \Psi_0)\phi^*_1\right] \\
    &=r_0'(\Psi_1, \Psi_0) \E\left[\phi^*_0\right] + r_1'(\Psi_1, \Psi_0) \E\left[\phi^*_1\right] \\
    &=0.
\end{split}
\end{align}
\end{proof}

\begin{lemma}\label{lemma:convergence_IF}
    Let the IF for the GLM based marginal effect estimator be given by $\phi_\Psi^*$ in \eqref{eq:IF}. We define $\phi_\Psi$ as the EIF for the marginal effect estimand, i.e. 
\begin{align} \label{eq:EIF}
\begin{split}
        \phi_{\Psi}(O) = r_0'(\Psi_1, \Psi_0)\phi_0(O) + r_1'(\Psi_1, \Psi_0)\phi_1(O),
\end{split}
\end{align}
with EIF for the counterfactual mean,
    \begin{align} \label{eq:EIF_mean}
    \begin{split}
        \phi_a(O) = \frac{1_a(A)}{\pi_a}(Y-\mu(W, a))+(\mu(W, a)-\Psi_a).
    \end{split}
    \end{align}
Then $\phi_\Psi^* \overset{L_2}{\to}\phi_\Psi$ is implied by $\mu^*(W, a) \overset{L^2}{\longrightarrow} \mu(W, a)$ for $n\to\infty$ and for $a\in\{0, 1\}$.
\end{lemma}

\begin{proof}
Assume $\mu^*(W, a) \overset{L^2}{\longrightarrow} \mu(W,a)$ for $n\to\infty$, then $\phi_a^* \overset{L_2}{\to}\phi_a$, since
\begin{align}
\begin{split}
    \E\left[\left(\phi_a^*(O)-\phi_a(O) \right)^2\right]=&\E\left[\left(\frac{1_a(A)}{\pi_a}(Y-\mu^*(W, a))+(\mu^*(W, a)-\Psi_a) - \left(\frac{1_a(A)}{\pi_a}(Y-\mu(W, a))+(\mu(W, a)-\Psi_a)\right)\right)^2\right] \\
    =&\E\left[\left(\frac{1_a(A)}{\pi_a}(\mu(W, a)-\mu^*(W, a))+\mu^*(W, a)-\mu(W, a)\right)^2\right] \\
    =&\E\left[\frac{1_a(A)}{\pi_a^2}(\mu(W, a)-\mu^*(W, a))^2\right]+\E\left[(\mu^*(W, a)-\mu(W, a))^2\right]\\
    &+ 2\E\left[\frac{1_a(A)}{\pi_a}(\mu^*(W, a)-\mu(W, a))(\mu(W, a)-\mu^*(W, a))\right]  \\
    =&\frac{1}{\pi_a}\E\left[(\mu(W, a)-\mu^*(W, a))^2\right]+\E\left[(\mu^*(W, a)-\mu(W, a))^2\right]\\
    &- 2\frac{1}{\pi_a}\E\left[(\mu^*(W, a)-\mu(W, a))^2\right] \to 0,
\end{split}
\end{align}
for $n\to\infty$. Now the convergence of the IF for the marginal effect estimator follows since, 
\begin{align}
\begin{split}
    \E\left[\left(\phi_\Psi^*(O)- \phi_\Psi(O)\right)^2\right]=&\E\left[\left(r_0'(\Psi_1, \Psi_0)\phi^*_0 + r_1'(\Psi_1, \Psi_0)\phi^*_1- r_0'(\Psi_1, \Psi_0)\phi_0 + r_1'(\Psi_1, \Psi_0)\phi_1\right)^2\right] \\
    =&\E\left[\left(r_0'(\Psi_1, \Psi_0)(\phi^*_0- \phi_0) + r_1'(\Psi_1, \Psi_0)(\phi^*_1-\phi_1)\right)^2\right] \\
    =&r_0'(\Psi_1, \Psi_0)^2\E\left[(\phi^*_0- \phi_0)^2\right] + r_1'(\Psi_1, \Psi_0)^2\E\left[(\phi^*_1-\phi_1)^2\right] \\
    &+ 2r_0'(\Psi_1, \Psi_0)r_1'(\Psi_1, \Psi_0) \E\left[(\phi^*_1-\phi_1)(\phi^*_0-\phi_0)\right].
\end{split}
\end{align}
The first two terms does converge to 0 for $n\to\infty$, as we just argued for. The last term converges to 0 since it holds that if $f_n\overset{L_2}{\to} f$ and $g_n\overset{L_2}{\to} g$ then $f_ng_n\overset{L_1}{\to} fg$.
\end{proof}

\subsection{Efficiency of GLM adjusted for true conditional control mean}

We are now ready to establish that under link-scale additive treatment assignment the GLM based plug-in estimator by \citet{RosenblumvanderLaan+2010} which adjusts for the link-scale true conditional control mean is efficient.

\begin{theorem}\label{theorem:true_prog}
Assume that the true treatment effect being additive on the link scale, $g\left(\mu(W, 1)\right) = \zeta + g\left(\mu(W, 0)\right)$ and let the true counterfactual mean function be denoted by $\mu(w,0)=\E[Y(0) | W= w]$. If we include $g(\mu(W,0))$ as an additional covariate in the GLM based plug-in procedure suggested by \citet{RosenblumvanderLaan+2010}, then $\mu^*(W, a) = \mu(W, a)$ for $a\in\{0, 1\}$.
\end{theorem}

\begin{proof}
Let $\beta_{\mu_0}^* = (\beta_0^*, \beta_a^*, \beta_g^*, \beta_w^*)$ maximize $\mathscr{L}(\mu_0, \beta)=\E_\mathcal{P}[\ell(O_{\mu_0}\, ; \beta)]$ over a compact parameter space $\{\beta : |\beta|\leq b\}$. The limiting GLM estimator where we use the true prognostic score $\mu(W, 0)$ uses the parameters $\beta_{\mu_0}^*$ which is the solution to $\mathscr{S}(\mu_0; \beta)=0$. The solution is always unique for the GLMs we consider, due to the strict concavity of the expected log-likelihood. In this proof we will show that the unique solution $\beta_{\mu_0}^*$ leads to $\mu^*(W, a) = \mu(W, a)$ for $a\in\{0, 1\}$.  

First, using \eqref{eq:consistency} derived from the consistency of the estimated counterfactual means $\widehat{\Psi}_a$, we have for $A=0$
\begin{align} \label{eq:oracle1}
\begin{split}
        & \E\left[\mu(W, 0)\right]= \E[Y(0)] =  \E[\mu^*(W, 0)]=\E\left[g^{-1}\left(\beta_0^*+g(\mu(W, 0))\beta_g^* + W\beta^*_w \right)\right]
\end{split}
\end{align}
and for $A=1$
\begin{align} \label{eq:oracle2}
\begin{split}
         \E\left[\mu(W, 1)\right]&= \E\left[g^{-1}\left(\zeta + g(\mu(w, 0))\right)\right] = \E[Y(1)] =  \E[\mu^*(W, 1)]\\
        &=\E\left[g^{-1}\left(\beta_0^*+ \beta_a^* + \left(g(\mu(W, 1))- \zeta\right)\beta_g^* + W\beta^*_w\right)\right]. 
\end{split}
\end{align}
In \eqref{eq:score_nb} the score for an distribution in the exponential family using the canonical link function is given. From this we can determine that $\mathscr{S}(\mu_0; \beta)=\E\left[\left(Y - \mu^*(W, A)\right)\cdot   \frac{\partial \eta}{\partial \beta}\right]$, where $\frac{\partial \eta}{\partial \beta} = \left(1, A, g(\mu(W, 0)), W\right)^\top$. From what is inside the expectation of \eqref{eq:oracle1} and \eqref{eq:oracle2} we see that if we choose $\beta_{\mu_0}^*=(0, \zeta, 1, 0)$ then $\mu^*(W, A)=\mu(W, A)$. Note that we are using that the treatment effect is additive on the link scale to conclude this. We will now show that the score equation is fulfilled when $\mu^*(W, A)=\mu(W, A)$. Thus, the $\beta_{\mu_0}^*$ that ensures $\mu^*(W, A)=\mu(W, A)$ is the unique solution to the score equation. Using the law of total expectation conditioning on $A$, the first score equation is 
\begin{align}
\begin{split}
    \E[Y-\mu^*(W, A)]=\E\left[Y-\mu^*(W, A)|A = 1\right]\pi_1 + \E\left[Y-\mu^*(W, A)|A = 0\right]\pi_0 . 
\end{split}
\end{align}
This is solved by $\mu^*(W, A)=\mu(W, A)$. The second score equation is also solved in this case, since 
\begin{align}
\begin{split}
    \E[A\left(Y-\mu^*(W, A)\right)]&=\E[AY(1)]-\E[A\mu^*(W, A)]\\
    &= \pi_1\E[Y(1)]-\E[A\E[Y|W, A]] \\
    &= \pi_1\E[Y(1)]-\E[\E[A Y(1)|W, A]]\\
    &= 0.
\end{split}
\end{align}
The third score equation is solved, since
\begin{align}
\begin{split}
    \E[g(\mu(W, 0))\left(Y-\mu^*(W, A)\right)]&= \E[g(\mu(W, 0))\left(Y-\mu(W, A)\right)]\\
    &= \E\left[\E\left[g(\mu(W, 0))\left(Y-\mu(W, A)\right)| W, A\right]\right]\\
    &= \E\left[g(\mu(W, 0))\E\left[Y-\mu(W, A)| W, A\right]\right]\\
    &= 0.
\end{split}
\end{align}
The fourth and last score equation is also solved using the same argument as for the third score equation. Thus, $\beta_{\mu_0}^*$ is the unique solution to the score equation. In other words, if we have the true conditional mean, including it in the link scale as a covariate ensures that the GLM simply "passes it through" as a prediction under these conditions. 
\end{proof}

Using this result we can conclude that, the IF in \eqref{eq:IF} equals the EIF in \eqref{eq:EIF}, thus we have established efficiency of the GLM based plug-in method using the true conditional control mean.

\subsection{Efficiency of GLM adjusted for emulated conditional control mean}\label{app:thm_main}
We are now ready to restate and proof \autoref{thm:main}. 
\restatedtheorem*

\begin{proof} 
First, let $\beta_{\widetilde{n}, \, n} = (\widetilde{\beta}_0, \widetilde{\beta}_a, \widetilde{\beta}_g, \widetilde{\beta}_w)$ maximize $\mathscr{L}_n(\hat\rho_{\tilde n}, \beta)= \frac{1}{n}\sum_{i=1}^n \ell(O_{\hat\rho_{\tilde n}}; \beta)$ and denote the conditional mean found by using the GLM model with $\beta_{\widetilde{n}, \, n}$ by
\begin{align*}
    \hat{\mu}_{\widetilde{n}, n}(W, A) = g^{-1}\left(\widetilde{\beta}_0+ \widetilde{\beta}_aA + g(\hat\rho_{\tilde n}\,(W))\widetilde{\beta}_g + W\widetilde{\beta}_w\right).
\end{align*}
Let $\beta_{\mu_0}^*$ maximize $\mathscr{L}(\mu_0, \beta)=\E_\mathcal{P}[\ell(O_{\mu_0}\, ; \beta)]$ over the compact parameter space $\{\beta : |\beta|\leq b\}$. Since we have assumed that $\mathcal{W}$ and $\mathcal{Y}$ are compact the Lipschitz condition is fulfilled and we can use \autoref{lemma:convergence_beta} to obtain $\beta_{\widetilde{n},\, n} \pto \beta_{\mu_0}^*$ for any $\hat\rho_{\tilde n}(W) \overset{L^2}{\longrightarrow} \mu(W, 0)$ when $\widetilde{n}, n \to \infty$. Since the parameter space is compact the convergence also hold in $L^2$, $\beta_{\widetilde{n},\, n} \overset{L_2}{\to} \beta_{\mu_0}^*$, which is a special case of Vitali convergence theorem.

Since $\hat\rho_{\tilde n}$ is a uniformly bounded random function and $\mu(W, 0)$ is bounded the canonical link function for the GLMs that we consider are all uniformly bounded with these as input. This implies that $\hat\rho_{\tilde n}(W) \overset{L^2}{\longrightarrow} \mu(W, 0) \implies g(\hat\rho_{\tilde n}(W)) \overset{L^2}{\longrightarrow} g(\mu(W, 0))$. Now we can use that $\mathcal{W}$ is compact to preserve to convergence in $L_2$ when including $g(\hat\rho_{\tilde n}(W))$ in the linear predictor of the GLM. Consequently using \autoref{theorem:true_prog} we have 
\begin{align} \label{eq:mu_conv}
\begin{split}
        \hat{\mu}_{\widetilde{n}, n}(W, A) \overset{L^2}{\longrightarrow} \mu(W, a),
\end{split}
\end{align}
when both $n$ and $\tilde{n}$ to go to $\infty$, which is obtained for $n=\mathcal{O}(\tilde{n})$ when $n\to\infty$, and when $\hat\rho_{\tilde n}(W) \overset{L^2}{\longrightarrow} \mu(W, 0)$ for $\tilde{n} \to \infty$.

\paragraph{}
The goal is now to show that the GLM based plug-in estimator $\hat{\Psi}$ in \eqref{eq:estimator} proposed by \citet{RosenblumvanderLaan+2010} that uses $g(\hat\rho_{\tilde n}\,(W))$ as an additional adjustment covariate is efficient under the assumption of additive treatment effect on the link scale. 
Now, we specifically want to show
\begin{align} \label{eq:goal}
\begin{split}
         \hat{\Psi}-\Psi= \frac{1}{n}\sum_{i=1}^n \phi_\Psi(O_i) + o_P(n^{-1/2}),
\end{split}
\end{align}
where $\hat{\Psi}-\Psi = \frac{1}{n}\sum_{i=1}^n \phi_\Psi(O_i) = o_P(n^{-1/2})$ is short-hand notation for $\sqrt{n}\left(\hat{\Psi}-\Psi- \frac{1}{n}\sum_{i=1}^n \phi_\Psi(O_i)\right) \overset{P}{\to} 0$. 
Now, fix the prognostic model $\hat\rho_{\tilde n}\,$ and include this in the GLM model on the link scale. Then $\hat{\mu}_{\widetilde{n}, n}(W, A)$ is the conditional mean found by using the GLM fit with $g(\hat\rho_{\tilde n}\,(W))$ included as an adjustment covariate, i.e. using $\beta_{\widetilde{n}, n}$. For a fixed $\widetilde{n}$ the limiting GLM conditional mean will be denoted by $\mu^*_{\widetilde{n}}(W, A)$, which we know fulfills \eqref{eq:consistency}. The estimator $\hat{\Psi}$ is RAL with IF $\phi_\Psi^*$ given in \eqref{eq:IF} using $\mu^*_{\widetilde{n}}(W, A)$, and we have 
\begin{align} 
\begin{split}
        \hat{\Psi}-\Psi&= \frac{1}{n}\sum_{i=1}^n \phi^*_\Psi(O_i) + o_P(n^{-1/2})\\
        &= \frac{1}{n}\sum_{i=1}^n \phi_\Psi(O_i) + \frac{1}{n}\sum_{i=1}^n \left(\phi^*_\Psi(O_i)-\phi_\Psi(O_i) \right)+ o_P(n^{-1/2}). 
\end{split}
\end{align}
This implies that $\hat{\Psi}$ is efficient if $\frac{1}{n}\sum_{i=1}^n \left(\phi^*_\Psi(O_i)-\phi_\Psi(O_i) \right) = o_P(n^{-1/2})$. Since convergence in $L_2$ implies convergence in probability, we can instead show  
\begin{align} 
\begin{split}
        n\E\left[\left(\frac{1}{n}\sum_{i=1}^n \left(\phi^*_\Psi(O_i)-\phi_\Psi(O_i) \right)\right)^2\right]\to 0.
\end{split}
\end{align}
An IF have mean zero when the
expectation is taken under the distribution to which they pertain. This implies that $\E[\phi_\Psi(O)]=0$. However, due to the consistency of $\hat{\Psi}_a$ it can be shown that $\E[\phi^*_\Psi(O)]=0$, see \autoref{lemma:mean_0_star}. Therefore, we have 
\begin{align} 
\begin{split}
        n\E\left[\left(\frac{1}{n}\sum_{i=1}^n \left(\phi^*_\Psi(O_i)-\phi_\Psi(O_i) \right)\right)^2\right]&= n\var\left(\frac{1}{n}\sum_{i=1}^n \left(\phi^*_\Psi(O_i)-\phi_\Psi(O_i) \right)\right)\\
        &= \frac{1}{n}\var\left(\sum_{i=1}^n \left(\phi^*_\Psi(O_i)-\phi_\Psi(O_i) \right)\right)\\
        &= \frac{1}{n}\sum_{i=1}^n\var\left(\phi^*_\Psi(O_i)-\phi_\Psi(O_i) \right)\\
        &= \var\left(\phi^*_\Psi(O)-\phi_\Psi(O) \right)\\
        &= \E\left[\left(\phi^*_\Psi(O)-\phi_\Psi(O) \right)^2\right],
\end{split}
\end{align}
where we are able to pass the variance through the sum in equality three since the prognostic model $\hat\rho_{\tilde n}$ is trained on historical data independent from the i.i.d trial data. The last equality implies that we just need $\phi_\Psi^* \overset{L_2}{\to}\phi_\Psi$ to complete the proof. In \autoref{lemma:convergence_IF} we show that $\phi_\Psi^* \overset{L_2}{\to}\phi_\Psi$ if $\mu^*_{\tilde{n}}(W, A) \overset{L^2}{\longrightarrow} \mu(W, a)$ under a known exposure model $\pi_a$. In \eqref{eq:mu_conv} we showed that $\hat\rho_{\tilde n}(W) \overset{L^2}{\longrightarrow} \mu(W, 0)$ gives $\hat{\mu}_{\widetilde{n}, n}(W, A) \overset{L^2}{\longrightarrow} \mu(W, a)$, when the treatment effect is additive on the link scale and for $n=\mathcal{O}(\tilde{n})$. However, we also know that $\hat{\mu}_{\widetilde{n}, n}(W, a) \overset{L^2}{\longrightarrow} \mu^*_{\widetilde{n}}(W, a)$ for $n\to \infty$. Using the squeeze lemma for convergence means that we are done.
\end{proof}

\section{. \hspace{0.1cm} Relative Efficiency of ATE Estimation with Nested Regressions in Randomized Trials}\label{app:ATE_nested}

We will use $\E_P[\cdot]$ to denote expectation with respect to a measure $P$. $\mathcal G^\perp$ denotes the orthogonal complement of a closed subspace $\mathcal G$ of $\mathcal L_2$. Let $\mathcal M$ be a set of measures over $(W, A, Y)$ such that we are in an 1:1 RCT setting, i.e. $A \sim \text{Bern}(1/2)$ independent of $W$. For any generic $P \in \mathcal M$, let $\mu_{P}(W, A) = \E_P[Y | W, A]$. We are interested in estimating the ATE evaluated at the true but unknown distribution $\mathcal{P}$, i.e. $\Psi = \Psi(\mathcal{P}) = \E_\mathcal{P}[\mu_{\mathcal{P}}(W, 1) -\mu_{\mathcal{P}}(W, 0)]$. Using the abbreviation $\mu(W, a) = \mu_{\mathcal{P}}(W, a)$, the estimand is $\Psi =\E_\mathcal{P} [\mu(W, 1)-\mu(W, 0)]$. We also use the abbreviation $\Psi_a = \E_\mathcal{P}[\mu(W, a)]$ for the treatment-specific mean at this distribution. Any RAL estimator of the ATE estimand has IF
$\phi^*(O) = \phi_{\mu^*,1}(O) - \phi_{\mu^*,0}(O)$ where
\begin{equation}\label{eq:IF_gen}
\phi_{\mu^*,a}(O) = \frac{1_a(A)}{1/2}(Y-\mu^*(W, a)) + (\mu^*(W, a) - \Psi_a)
\end{equation}
for some function $\mu^*(W, a)$ such that $\E_\mathcal{P}[\mu^*(W, 1) - \mu^*(W, 0)] = \Psi$. Notice that now $\mu^*$ does not have to be the large sample GLM model that we have otherwise considered in this study. The EIF is now $\phi(O) = \phi_{1}(O)-\phi_{0}(O)$ using \eqref{eq:EIF_mean}. The asymptotic variance of the RAL estimator with IF $\phi^*$ is $\|\phi^*\|_{L_2}^2$. From now on, when we write $\E[\cdot]$ we mean $\E_\mathcal{P}[\cdot]$. We will not be taking expectations under any distributions other than $\mathcal{P}$ so this will not introduce any ambiguity.

Assume $\overline\mu(\cdot, a)$ are the population minimizers of mean-squared error in some class of functions $\overline{\mathcal M}$ which is a closed subspace of $\mathcal L_2$. Let $\widetilde\mu(\cdot, a)$ be the equivalent for a larger closed subspace $\widetilde{\mathcal M} \supset \overline{\mathcal M}$. For example, imagine this is the limit of two regression where $\widetilde\mu(\cdot, a)$ uses a \textit{larger} set of basis functions for $W$ (e.g. including an estimated prognostic score). Our high-level question is whether adding terms in the regression always reduces the asymptotic variance. Concretely, abbreviating $\widetilde\phi = \phi_{\widetilde\mu}$ and $\overline\phi = \phi_{\overline\mu}$,  is $\|\widetilde\phi\|_{L_2} \le \|\overline\phi\|_{L_2}$? Indeed, a special case of this was proven by \citet{Schuler2020}, who showed that when $\pi_1=\pi_0$ or when there is a constant treatment effect adding covariates to the linear regression will only increase the asymptotic efficiency of the ATE estimate. However, when $\pi_1\neq \pi_0$ \textit{and} there is a heterogeneous treatment effect there exists cases where adding covariates will result in a higher asymptotic variance compared to the difference-in-mean estimator. In this section we are assuming 1:1 randomization and hence $\pi_1=\pi_0$, so the result does indeed hold for the linear model in the scenario that we are considering.

\subsection{Influence Function Geometry} 
By the geometry of influence functions in this setting, the difference $\phi^* -\phi$ is orthogonal to $\phi$ for all influence functions $\phi^*$ (see \autoref{fig:if}) \cite{tsiatis2007semiparametric}. Therefore by the Pythagorean Theorem the norm of $\widetilde\phi$ is less than that of $\overline\phi$ if and only if $\widetilde\phi$ is closer to the efficient IF than $\overline\phi$ is, i.e. if $\| \widetilde\phi - \phi \|_{L_2} \le \| \overline\phi - \phi \|_{L_2}$. The utility of this is that terms like $\|\phi^* - \phi\|_{L_2}$ can be expressed quite simply.

\begin{figure}[ht!]
\centering
\begin{tikzpicture}[scale=3,tdplot_main_coords]
    \def\x{.5}
    \filldraw[
        draw=teal,
        fill=teal!20,
    ]          (-1.5,-0.5,-1)
            -- (1,-0.5,-1)
            -- (1,1,-1)
            -- (-1.5,1,-1)
            -- cycle; 
    \node[teal] at (-0.3,-1.4,-1) {Influence Functions of $\Psi$ at $\mathcal{P}$};

    \coordinate (origin) at (0,0,0);
    
    \coordinate (EIF) at (0,0,-1);
    \draw[blue, thick, <->] (0,0,-1.1) -- (0,0,0.3) node[right] {Tangent Space of $\mathcal M$ at $\mathcal P$}; 

    \coordinate (tilde-phi) at (0.5,0.7,-1);
    \fill (tilde-phi) circle (0.5pt) node[right] {$\widetilde\phi$};
    \draw[thin,->] (origin) -- (tilde-phi) node[] {}; 

    \coordinate (bar-phi) at (-1,0.75,-1);
    \fill (bar-phi) circle (0.5pt) node[right] {$\overline\phi$};
    \draw[thin,->] (origin) -- (bar-phi) node[] {}; 
    
    \fill (EIF) circle (0.5pt) node[left] {$\phi$};
    \fill (origin) circle (0.5pt) node[left] {$0$};
    
    \draw[gray, dashed, thin,->] (EIF) -- (bar-phi) node[midway, above] {$\overline\phi-\phi$}; 
    \draw[gray, dashed, thin,->] (EIF) -- (tilde-phi) node[midway, below left] {$\widetilde\phi-\phi$}; 
    \end{tikzpicture}
\caption{Conceptual illustration of various influence functions. All influence functions (up to an offset of $\phi$) are orthogonal to the tangent space of $\mathcal M$ at $\mathcal P$.}
\label{fig:if}
\end{figure}

\begin{lemma}
\label{thm:if-norm-diff}
Abbreviate $\phi^* = \phi^*_\mu$ for some allowable $\mu^*$. In our setting, $\|\phi^* - \phi\| = \|\omega^* -\omega \|$ where $\omega^*(W) = \mu^*(W, 1) + \mu^*(W, 0)$ and $\omega(W) = \mu(W, 1) + \mu(W, 0)$.
\end{lemma}

\begin{proof}
Firstly, we can decompose the difference as,
\begin{align*}
\phi^* - \phi
&= (\phi_1^* - \phi_0^*) - (\phi_1 - \phi_0)\\
&= (\phi_1^* - \phi_1) - (\phi_0^* - \phi_0).
\end{align*}
Using \eqref{eq:IF_gen}, we have
\begin{align*}
\phi_{a}^*(W, A, Y) - \phi_{a}(W, A, Y) 
&= 2 \cdot 1_a(A)\bigl(\mu(W, a)-\mu^*(W, a)\bigr) + \bigl(\mu^*(W, a)-\mu(W, a)\bigr) \\
&= \bigl(1 - 2\cdot1_a(A)\bigr)\, \bigl(\mu^*(W, a)-\mu(W, a)\bigr) \\
&= Z_a \bigl(\mu^*(W, a)-\mu(W, a)\bigr),
\end{align*}
where $Z_a = 1 - 2\cdot1_a(A)$ is a symmetric Bernoulli variable which takes value $-1$ when $A=a$ and $1$ otherwise. By definition, $Z_1 = -Z_0$. Since they depend only on $A$, the variables $Z_a$ are independent from $W$. Thus, 
$
\phi^* - \phi
= Z_1(\mu^*(W, 1)-\mu(W, 1)) - Z_0 (\mu^*(W, 0)-\mu(W, 0)) = Z_1(\omega^* - \omega)
$. Thus,
\begin{align*}
\|\phi - \phi^* \|_{L_2}^2
&= \E\left[Z_1^2\bigl(\omega(W) - \omega^*(W)\bigr)^2\right] \\
&= \E\left[Z_1^2\right] \E\left[\bigl(\omega(W) - \omega^*(W)\bigr)^2\right] \\
&= \|\omega - \omega^* \|_{L_2}^2 \\
\end{align*}
\end{proof}

\begin{corollary}
\label{thm:if-norm-diff-2}
Let $f(w, a), g(w, a)$ index $\phi_f, \phi_g$. Then $\|\phi_g\|^2 - \|\phi_f\|^2 = \|\omega_g-\omega\|^2 - \|\omega_f-\omega\|^2$.
\end{corollary}
\begin{proof}
For any $\phi^*$, $\|\phi^*\|_{L_2}^2 = \|\phi^* - \phi\|_{L_2}^2 + \|\phi\|_{L_2}^2$ since $\phi \perp (\phi^* - \phi)$. Applying this to $\phi_f$ and $\phi_g$, taking the difference, and applying \autoref{thm:if-norm-diff} gives the result.
\end{proof}

\autoref{thm:if-norm-diff-2} says that the difference in efficiency depends on the geometry of the triangle created by $(\omega, \omega_g,\omega_f)$. If the side between $\omega$ and $\omega_f$ is longer than the side between $\omega$ and $\omega_g$ then the estimator using $f$ as the estimated mean function is more efficient.

\subsection{Nested Regressions}

Return now to the setting where $\overline\mu$ and $\widetilde \mu$ are the limits in the ``nested'' regression models $\overline{\mathcal M} \subseteq \widetilde{\mathcal M}$. The key observation is that, by virtue of minimizing mean-squared error, $\widetilde\mu$ is the unique $\mathcal L_2$ projection of $\mu$ onto $\widetilde{\mathcal M}$ and $\overline\mu$ is the equivalent projection onto $\overline{\mathcal M}$. Moreover, $\overline\mu$ is also the projection of $\widetilde\mu$ onto $\overline{\mathcal M}$ by using a iterated projection. Because of this we know $\widetilde\mu(\cdot, a) - \mu(\cdot, a) \in \widetilde{\mathcal M}^\perp$,  $\overline\mu(\cdot, a) - \mu(\cdot, a) \in \overline{\mathcal M}^\perp$, and  $\overline\mu(\cdot, a) - \widetilde\mu(\cdot, a) \in \overline{\mathcal M}^\perp \cap \widetilde{\mathcal M}$. Since $\mathcal L_2$ projection is a linear operator, the same relationships hold between the respective $\omega$ functions. \autoref{fig:proj} illustrates the geometry conceptually.

\begin{figure}[ht!]
\centering
\begin{tikzpicture}[scale=3,tdplot_main_coords]
    \def\x{.5}
    \filldraw[
        draw=teal,
        fill=teal!20,
    ]          (-1.5,-0.5,0)
            -- (1,-0.5,0)
            -- (1,1,0)
            -- (-1.5,1,0)
            -- cycle; 
    \node[teal] at (-0.3,1.1,0) {$\widetilde{\mathcal M}$};

    \draw[blue, thick,<->] (-1.5,0,0) -- (1,0,0) node[anchor=north east]{$\overline{\mathcal M}$};
    \draw[thick,<->] (0,-0.5,0) -- (0,1,0) node[anchor=north west]{};
    \draw[thick,<->] (0,0,-0.5) -- (0,0,1) node[anchor=south]{};

    \coordinate (mu-star) at (-1,0.5,1);
    \fill (mu-star) circle (0.5pt) node[right] {$\mu$};
    \coordinate (mu-tilde) at (-1,0.5,0);
    \fill (mu-tilde) circle (0.5pt) node[right] {$\widetilde\mu$};
    \coordinate (mu-bar) at (-1,0,0);
    \fill (mu-bar) circle (0.5pt) node[left] {$\overline\mu$};

    \draw[gray, dashed, thin,->] (mu-star) -- (mu-tilde) node[midway, right] {$\widetilde\mu - \mu$}; 
    \draw[gray, dashed, thin,->] (mu-star) -- (mu-bar) node[midway, left] {$\overline\mu - \mu$}; 
    \draw[gray, dashed, thin,->] (mu-tilde) -- (mu-bar) node[midway, below] {$\overline\mu - \widetilde\mu$}; 
\end{tikzpicture}
\caption{Conceptual illustration of the limiting regressions $\overline\mu$ and $\widetilde\mu$ as projections of $\mu$ onto their respective nested models.}
\label{fig:proj}
\end{figure}

\nested*

\begin{proof}
Let $\widetilde\omega = \widetilde\mu(\cdot, 1) + \widetilde\mu(\cdot, 0)$ and similar for $\overline\omega$. By \autoref{thm:if-norm-diff-2}, $\widetilde\phi$ has smaller norm if and only if $\|\widetilde\delta\| = \|\widetilde\omega - \omega\| < \|\overline\omega - \omega\| = \|\overline\delta\|$. The vectors $\widetilde\delta$, $\overline\delta$, and $\alpha = \overline\omega - \widetilde\omega$ form the sides of a triangle, see \autoref{fig:proj}. Furthermore, we have $\widetilde\delta = \widetilde\omega - \omega \in \widetilde{\mathcal M}^\perp$, as concluded before \autoref{thm:nested}. However, $\alpha = \overline\omega - \widetilde\omega \in \widetilde{\mathcal M}$, so $\alpha$ and $\widetilde\delta$ are orthogonal and thus are the two ``sides'' of a right triangle, while $\overline\delta$ is the hypotenuse. The hypotenuse cannot be shorter than the sides by the Pythagorean theorem $\|\overline\delta\|^2 = \|\widetilde\delta\|^2 + \|\alpha\|^2$. The inequality is strict since $\|\alpha\| = \|\overline\omega-\widetilde\omega\| > 0$ by our assumption that $\overline{\mathcal M}$ is a strict subset of $\widetilde{\mathcal M}$.
\end{proof}

The result says, for example, that adding terms to a linear regression can only \textit{improve} the efficiency of the resulting plug-in ATE estimator in a randomized trial setting with a 1:1 randomization scheme.

\section{. \hspace{0.1cm} Results in different data generation scenarios} \label{app:sim_dpg_cv_no}

\begin{figure}[!ht]
    \centering
    \includegraphics[width=1\textwidth, clip]{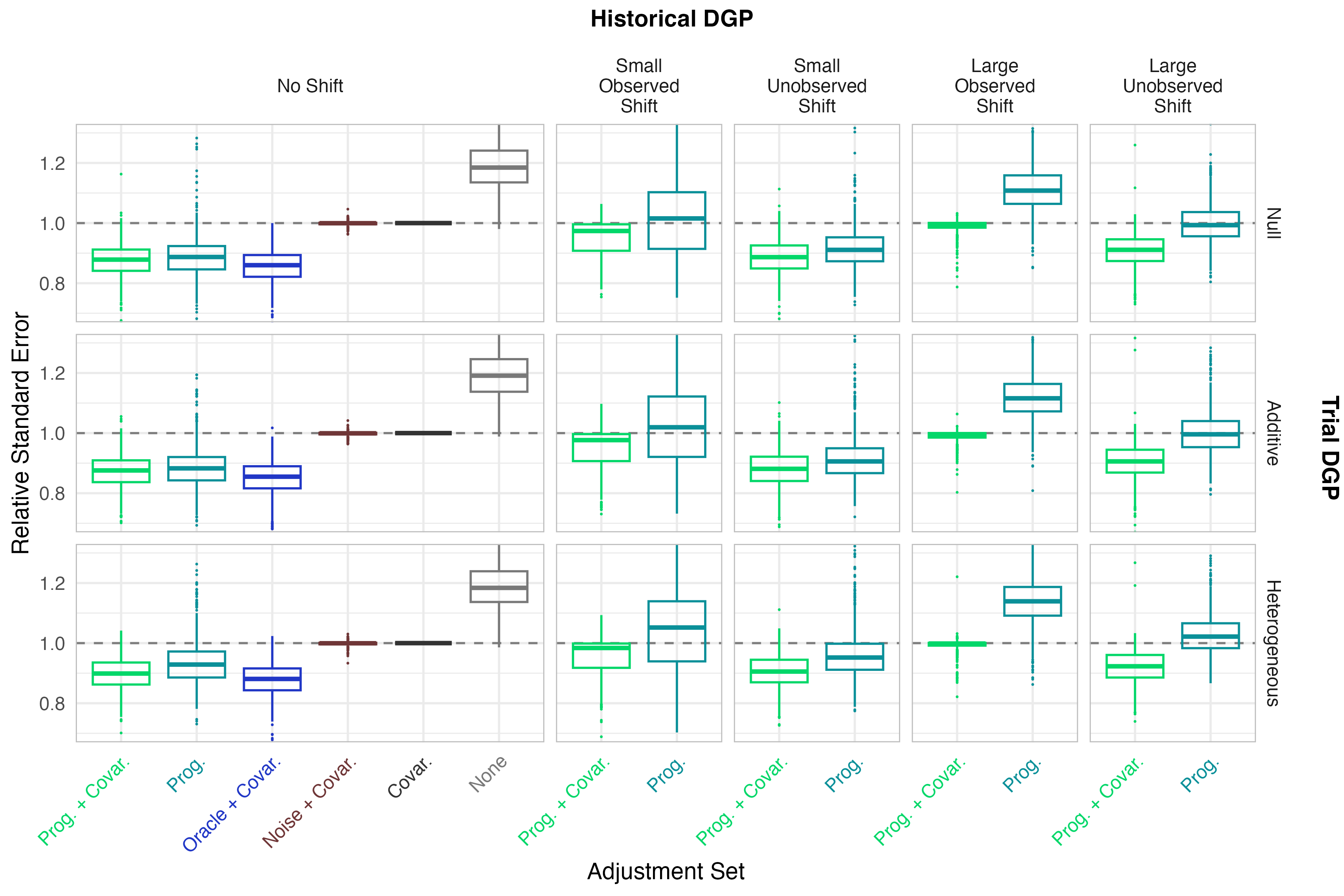}
    \caption{Equivalent of \autoref{fig:se}, but with standard errors estimated without cross-fitting.}
    \label{fig:se-nocv}
\end{figure}

\begin{figure}[!ht]
    \centering
    \includegraphics[width=1\textwidth, clip]{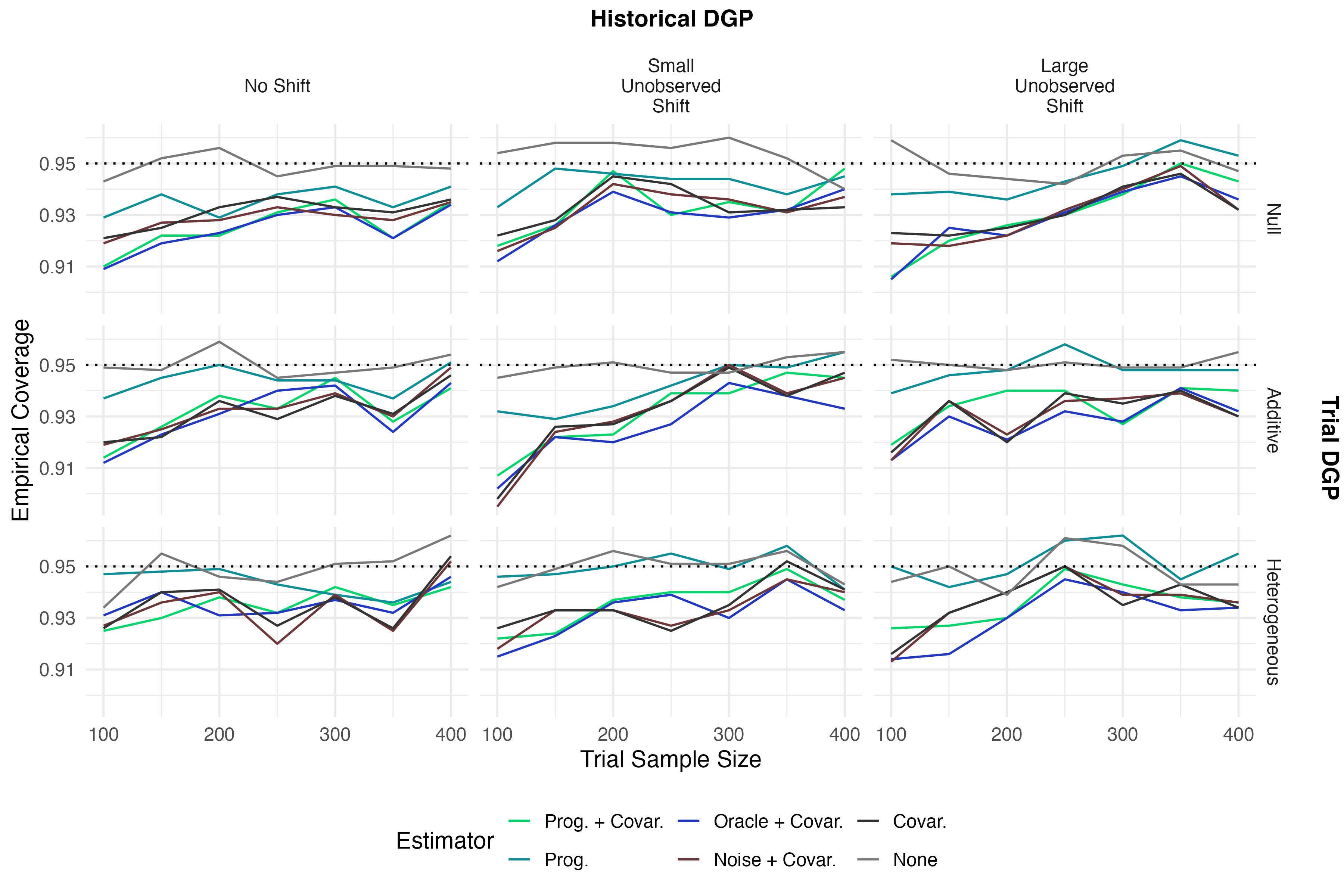}
    \caption{Equivalent of \autoref{fig:coverage}, but with standard errors estimated without cross-fitting.}
    \label{fig:coverage-nocv}
\end{figure}

\begin{figure}[!ht]
    \centering
    \includegraphics[width=1\textwidth, clip]{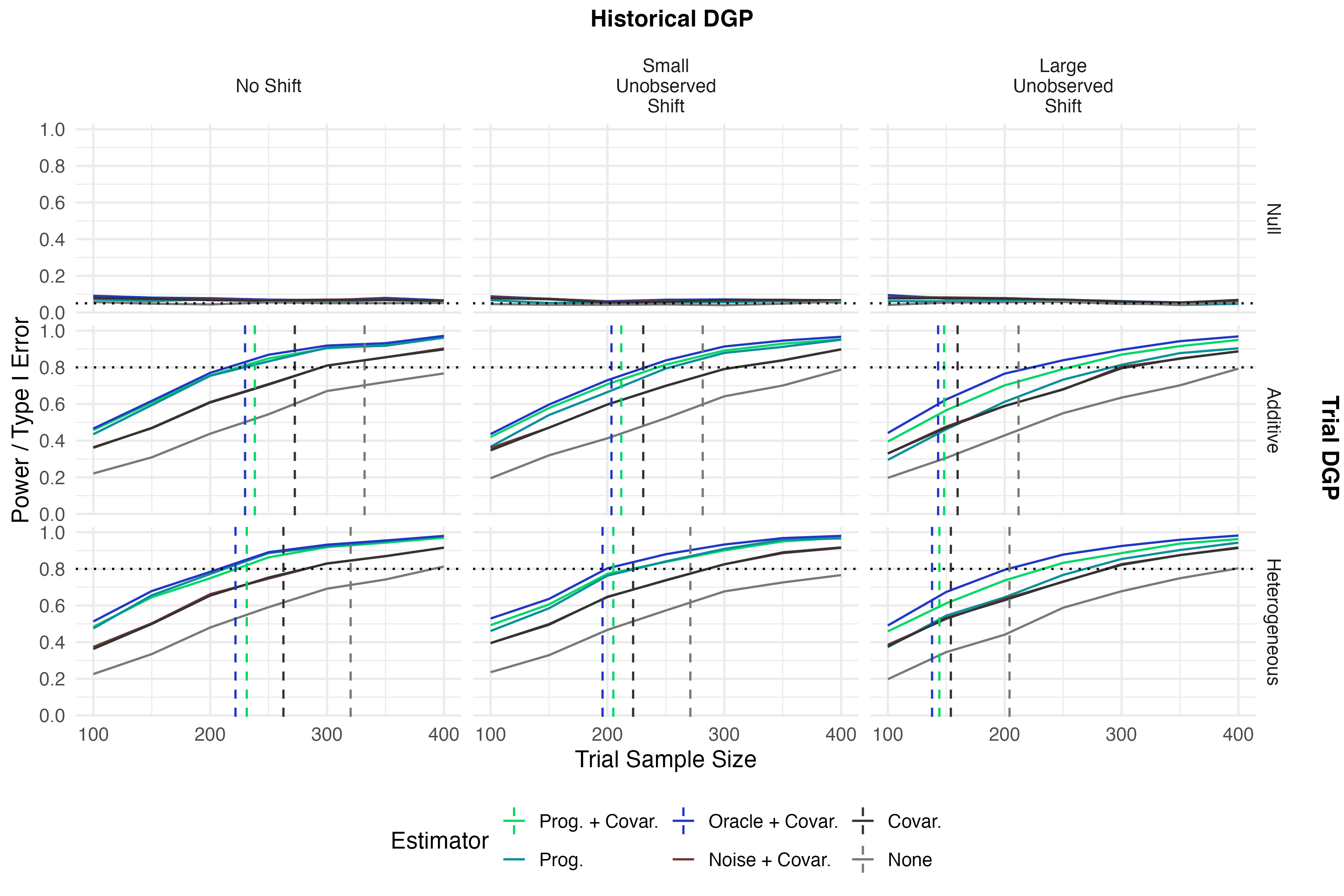}
    \caption{Equivalent of \autoref{fig:se}, but with standard errors estimated without cross-fitting.}
    \label{fig:power-nocv}
\end{figure}

\FloatBarrier
\section{. \hspace{0.1cm} Case study}\label{app:case_study}

\subsection{ \hspace{0.1cm} Summary of case study data} \label{app_overview_cs}

\setlength\heavyrulewidth{0.50ex}
\begin{table}[hb!!]
\centering
\begingroup
{
\renewcommand*{\arraystretch}{0.8}
\begin{tabular}{llrrrrr}
\toprule
\textbf{Data name} & \textbf{Trial ID} & \textbf{Duration} & \textbf{Titration target} & \textbf{Blinding type} & \multicolumn{2}{l}{ \textcolor{white}{test} \underline{ \textcolor{white}{test}\textbf{Number of participants}\textcolor{white}{test}}}\\ 
  & & & (mmol/L)& & \\
\multicolumn{5}{r}{} & \textcolor{white}{test} Randomized & Completed \\
  \midrule
\rowcolor{lightgray} Trial & NN9068-4229 & 26 weeks & 4.0-5.0 & Open-label & 210 & 206 \\
 & NN9068-4228 & 104 weeks & 4.0-5.0 & Open-label & 504 & 481 \\ 
 & NN1250-3579 & 52 weeks  & 4.0-5.0 & Open-label & 257 & 197 \\
 & NN1250-3586 & 26 weeks  & 4.0-5.0 & Open-label & 146 & 136\\
 & NN1250-3672 & 26 weeks & 4.0-5.0 & Open-label & 230  & 201 \\
 Historical & NN1250-3718 & 26 weeks & 4.0-5.0 & Open-label & 234  & 209 \\
 & NN1250-3724 & 26 weeks & 4.0-5.0 & Open-label & 230  & 206 \\
 & NN1250-3587 & 26 weeks & 4.0-5.0 & Open-label & 278  & 254 \\
 & NN5401-3590 & 26 weeks & 3.9-5.0 & Open-label & 264  & 232 \\
 & NN5401-3726 & 26 weeks & 3.9-5.0 & Open-label & extension of 3590 & 209\\ 
 & NN5401-3896 & 26 weeks & 3.9-5.0 & Open-label & 149 &  137\\

\bottomrule
\end{tabular}
}
\caption{Summary of case study data provided by Novo Nordisk A/S. The trial data set used for analysis of the RR is highlighted in grey. The historical data consists of all the data sets that are not highlighted. The number of participants refers to the number of participants receiving insulin IGlar. }\label{tab:cs_trials}
\endgroup
\end{table}

\FloatBarrier

\subsection{ \hspace{0.1cm} Data missingness} \label{app_datmis}

There was no missing data for the primary endpoint i.e. the number of ADA classified hypoglycaemic episodes for any of the 10 trials. 

A total of $95\%$ of the participants had complete data in the pooled trial and historical data for the baseline covariates. A total of 49 baseline covariates was included in the data, for more details on the specific covariates see \citet{Liao2023}. A missingness pattern plot for the covariates can be seen in \autoref{fig:missingness}. The missing covariates for the remaining $5\%$ of the participants were imputed using an RF on the historical and testing data, separately. In the historical data sample, the normalized root mean square error for continuous covariates was 0.211, with a false classification rate of 0.003. For every covariate with missing values, a missingness indicator was created as an additional covariate for prognostic model development {\color{black}(this is a standard approach for missing covariates in predictive models \cite{groenwold2020informative})}. In contrast, the new trial data exhibited a normalized root mean square error of 0.11, with no instances of falsely classified data.

\begin{figure}[!ht]
    \centering
    \includegraphics[width=1\textwidth, trim={0mm 0mm 0mm 0mm}, clip]{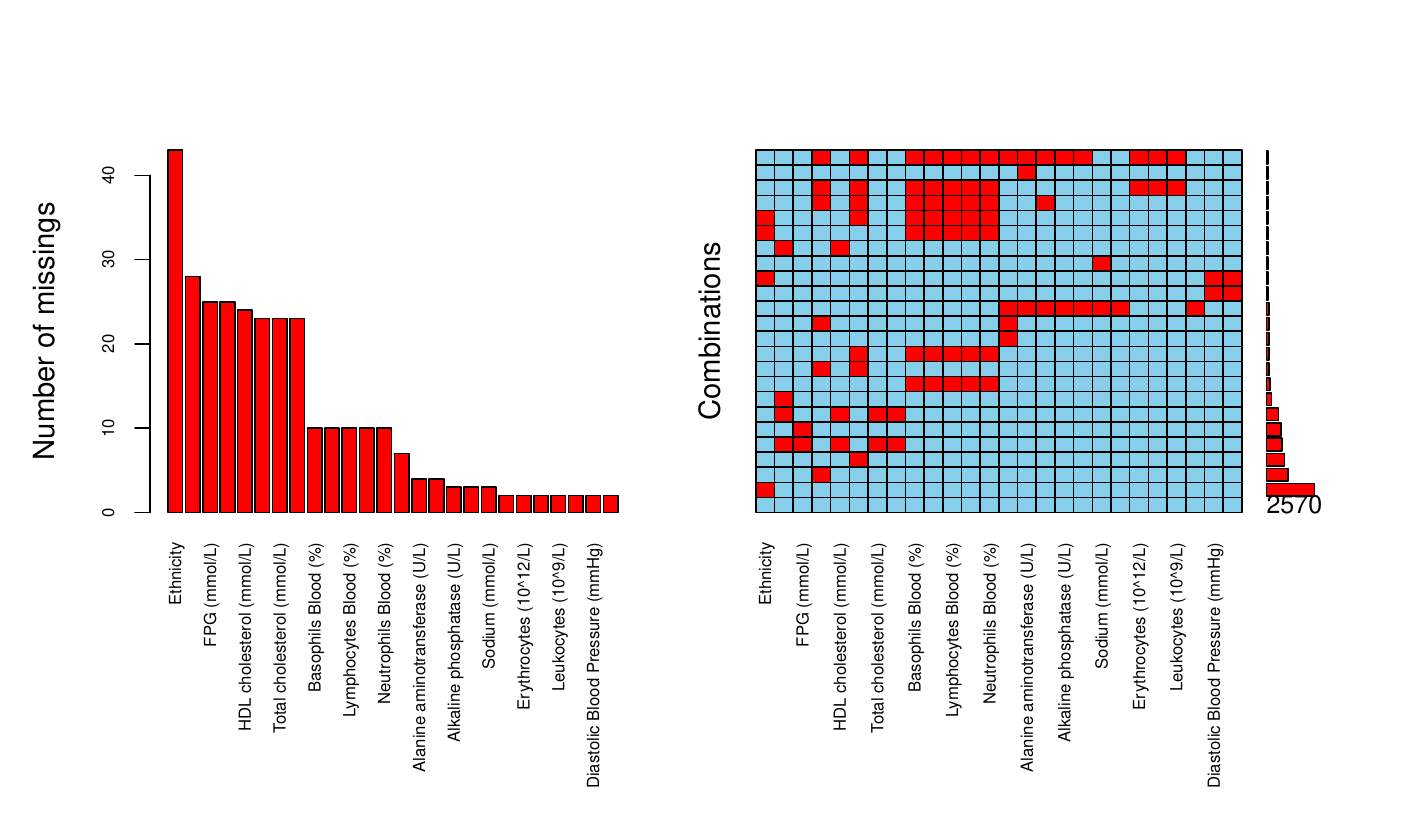}
    \caption{Left: Total number of missing values for each covariate. Right: Combination pattern of missingness.}
    \label{fig:missingness}
\end{figure}

\subsection{Discrete Super learner} \label{app:case_tuning}

\textbf{Number of folds}: Cross-validation is used to select the best candidate learner in the library for historical sample. A 5-fold scheme was used. \\
\textbf{Library of learners}: 
\begin{itemize}
    \item Multivariate Adaptive Regression Splines with the highest interaction to be to the 3rd degree
    \item Poisson regression
    \item Extreme gradient boosting with specifications: learning rate 0.1, tree depth 3, crossed with trees specified 25 to 500 by 25 increments
    \item Random Forest with number of trees found from cross validation  specified by 25 to 500 by 25 increments.
    \item K nearest neighbors with number of neighbors between 3, 4, 5, 7 and 9 found by cross validation
    \item Lasso Poisson regression with penalty found by cross validation
\end{itemize}
\textbf{Loss function}: Mean square error loss.

\end{document}